\documentclass{amsart}
\advance\oddsidemargin by -1.0cm
\advance\evensidemargin by -1.0cm
\usepackage{amsmath}%
\usepackage{amsfonts}%
\usepackage{amssymb}%
\usepackage{graphicx}
\usepackage[latin1]{inputenc}               
\newtheorem{thm}{Theorem}
\theoremstyle{plain}

\newtheorem{cor}{Corollary}

\newtheorem{dfn}{Definition}

\newtheorem{lem}{Lemma}

\newtheorem{rem}{Remark}

\numberwithin{equation}{section}
\newcommand{\bdm}{\begin{displaymath}}
\newcommand{\edm}{\end{displaymath}}
\newcommand{\Ric}{{\mathrm{Ric}}}
\newcommand{\x}{\times}

\newcommand{\ra}{\rightarrow}
\newcommand{\lra}{\longrightarrow}

\newcommand{\SO}{\ensuremath{\mathrm{SO}}}
\newcommand{\R}{\ensuremath{\mathbb{R}}}
\newcommand{\N}{\ensuremath{\mathbb{N}}}
\newcommand{\U}{\ensuremath{\mathrm{U}}}
\newcommand{\Lie}{\ensuremath{\mathcal{L}}}           
\begin{document}

\title{Geometry of higher dimensional black holes}
\author{Bernadette Lessel}
\address[Bernadette Lessel]{Mathematisches Institut, Universität Göttingen}%
\email{blessel@uni-math.gwdg.de}%
\thanks{This work originated from a diploma thesis submitted in September 2011 at Marburg University. I would like to thank Prof. Dr. Ilka Agricola of the Institut für Mathematik und Informatik of the Marburg University. Without her help and guidance this work would not have been possible.}
\date{October 23, 2012}

\begin{abstract}
This article investigates higher dimensional vacuum solutions of the Einstein equations.
Generalizations of the definitions of spherical and axial symmetry to higher dimensions are discussed before 
analyzing specific solutions bearing one of these symmetries. The effective motions of the Tangherlini metric are 
calculated and its Kruskal continuation is derived. Also the Myers-Perry metric is analyzed with respect to its causal 
and horizontal structure.
\end{abstract}
\maketitle

\section{Introduction}

The fact that there are only three dimensions of space is an assumption about nature that was implicitly implemented in physical 
theories. By now we have no deeper theory wich determines the dimensionality of spacetime. So it is reasonable to investigate the 
question which special features the dimension $3+1$ has by means of the theories we assume to be true. This article is an attempt to 
contribute to this question by studying higher dimensional General Relativity.

The first one who thought about studying higher dimensional vacuum solutions of the Einstein equations for this reason was Tangherlini 
in the 1960's \cite{tangh} where he found the unique static and spherical symmetric solution for arbitrary spactime dimension, now called 
Tangherlini metric. In 1986 Myers and Perry found a new family of solutions \cite{myersperry} which describe rotating black holes 
in higher dimensional spacetimes and can be considered as a generalization of the Tangherlini metric to the non static case.
Progress ist also recently made by Emparan, Reall et al \cite{empmat}, \cite{rotri}, \cite{bring}, \cite{emprl} who among other things showed 
that axial symmetric vacuum solutions need not to be unique in spacetime-dimension higher than four.

Some people see additional motivation for studying this topic by hoping to find possible factors of a higher dimensional solution of 
Superstring theory.

This article is organized as follows:
The first section establishes the notions for arbitrary-dimensional generalization of the four dimensional spherical and axial symmetry. 
The following two sections analyze solutions of this kind of symmetry. Section two investigates the Tangherlini metric which can be seen as an arbitrary-dimensional generalization of the Schwarzschild metric. For this metric the effective potential is calculated and snapshots of numerical simulations of it were added. Furthermore, a Kruskal continuation for this metric is derived. It can be seen that for spherical symmetry the unique solution behaves quite simmilarly in every dimension. The appendix of the paper derives the Ricci flatness of the Tangherlini metric, which apparently cannot be found in the literature by now. 
Section two investigates the (non-unique) axisymmetric case. To understand the issue properly, we begin with the four-dimensional case, namely the Kerr metric, and recall its causal structure. After this its generalization, the Myers-Perry solutions, are discussed in detail especially its horizon and causality structure. At first we describe rotation in just one plane, then we proceed with rotation in every possible direction.
We close the paper with a discussion of the horizon functions where we relate the different horizon generating functions of the different 
metrics with each other and find out that they have a suprisingly simple mathematical form, namely that they are "similar" to polynomials.

\section{Spacetime symmetries}

Convenient spacetime symmetries for General Relativity are the \emph{spherical} and \emph{axial} symmetry. The famous 
Schwarzschild and Kerr solution of the Einstein equation 
\bdm
\Ric-\frac{1}{2}R\cdot g=T,
\edm
either bear one of these symmetries. A natural question ist thus, how to formulate these symmetries for higher dimensional 
spacetimes. This is what will be tackled in this section. Before that, we will lay our eyes on two other notions, which are 
also very important. Namely the \emph{stationary} and the \emph{static} spacetime. For this, let in the whole section $(M,g)$ 
be a Lorentz manifold with signature $(-,+,...,+)$.

\begin{dfn}
$(M,g)$ will be called
\begin{enumerate}
	\item[1)] \emph{stationary}, if there exists a timelike Killingvector $K$ on $M$.
  \item[2)] \emph{static}, if it is stationary and $K^\bot$ is integrable.
\end{enumerate}
\end{dfn}

\begin{rem}
\begin{itemize}
\item[] 
\item With Frobenius' Theorem a Lorentz manifold is static if and only if for $\omega:=K^\flat$ it holds 
          $\omega\wedge d\omega=0$.
\item To every point of a static manifold there exist an open neighbourhood with coordinates $\{(t,x^i)\}$ in which the
          metric takes the form 
          \bdm
          g=g_{00}(x)dt\otimes dt+g_{ij}(x)dx^i\otimes dx^j,
          \edm
          where $g_{00}=g(K,K)$. For a proof of this statement see \cite{strau}.

\end{itemize}
\end{rem}

We will now focus our attention on the spherical symmetry. At first we will consider this notion at the familiar level of 
four dimensions. 

\begin{dfn}
A four-dimensional Lorentz manifold $(M,g)$ is called \emph{spherical}, if there exists a group action $L_A:M\lra M$, 
$A\in\SO(3)$, of $\SO(3)$ onto the manifold $M$, such that $L^*_A g=g$ $\forall A\in\SO(3)$ and every orbit is a two-dimensional 
spacelike surface. 
\end{dfn}

In what follows we consider a static spherical symmetric manifold with a unique Killingvector. The additional assumptions 
allow the formulation of the following statement.

\begin{lem}\label{spher}
Let $(M,g)$ be a manifold with the above assumptions. Then locally the metric $g$ can be written as 
\bdm
g=-e^{2a(r)}dt\otimes dt+e^{2b(r)}dr\otimes dr+ r^2g_{S^2},
\edm
where $t\in\R$, $r\in (R,\infty)$, $R\in\R^+$ and $g_{S^2}$ the Riemann metric on the sphere.                              
\end{lem}
For a proof of this see again \cite{strau}. It is a well known theorem by Birkhoff that says that every spherical symmetric
manifold is automatically static.                                                                                                

\noindent 
We are now prepared for the definition of a static and spherical symmetric arbitrary-dimensional Lorentz manifold, since      
we use for this generalization the result of Lemma \ref{spher}.

\begin{dfn}\label{dfn.sph}
We call a $d+1$-dimensional Lorentz manifold $(M,g)$ \emph{static and spherical symmetric}, if locally $g$ can be written 
in the form 
\bdm
g=-e^{2a(r)}dt\otimes dt+e^{2b(r)}dr\otimes dr+ r^2g_{S^{d-2}},
\edm
where $g_{S^{d-2}}$ is the Riemann metric of the $d-2$-sphere. 
\end{dfn}

\begin{rem}
 \begin{itemize}
 \item[]
 \item The Riemann metric $g_{S^{n}}$ of the $n$-sphere with radius $1$ is of the shape
         \bdm
          g_{S^n}=\sum_{k=1}^n\left(\prod_{s=1}^{k-1}\sin^2\chi_s\right)\ d\chi_k\otimes d\chi_k,
         \edm
          for $n\in\N$ and where $\{\chi_i\}$, $i=1,...,n$, are the $n$-dimensional spherical coordinates.                         
          Put thereby for the empty product $\prod_{s=1}^0\sin^2\chi_s:=1$. In particular for $g_{S^2}$ it holds
         \bdm
         g_{S^2}=d\theta\wedge d\theta+ \sin^2\theta\ d\varphi\wedge d\varphi,
         \edm
         for $\theta=\chi_1$, $\varphi=\chi_2$. 
 \item It is supposed to hold that the above metric bears the most general shape of a metric on a $d+1$-dimensional
          stationary manifold allowing $\SO(d-1)$ as isometriegroup. Anyway, a proof is not known to the author.
 \end{itemize} 
\end{rem}

We will now have a look at how axial symmetry can be generalized to arbitrary dimensional spacetimes. 

\begin{dfn}
$(M,g)$ is called \emph{stationary and axial symmetric}, if the group $\R\x\U(1)^{d-2}$ acts isometrically, in a way that the 
orbits of the action of $\U(1)^{d-2}$ are spacelike. Additionally it is required that the Killingfield belongig to the action 
of $\R$ is asymptotically timelike.
\end{dfn}

\begin{rem}
 \begin{itemize}
 \item[]
 \item For $d=3$, $(M,g)$ is stationary and axial symmetric iff $\R\x\U(1)$ acts isometrically. Because of $U(1)\cong\SO(2)$          
           the previously given definition is indeed a generalization of the fourdimensional axial symmetry. Graphically 
           spoken, in our generalized definition we don't just consider one rotation around one axis, but $d-2$ rotations
           around spacelike hypersurfaces of codimension $2$.
 \item It is also possible to generalize the fourdimensional axial symmetry in a way that is demanded that the group $\SO(d-1)$
           acts isometrically in such a way, in that the orbits are spacelike $(d-2)$-dimensional spheres. But for the extraction of 
           solutions to the Einstein equation the above given definition is more practicable.
 \item Our definition of higher dimensional axial symmetry however has one limitation. Namely only in dimensions $4$ and $5$ 
           there exist axial symmetric manifolds which are asymptotically the Minkowskispace (that means, which are asymptotically flat)  
           and in this sense are physically significant. 
 \end{itemize}          
\end{rem}

The following theorem of T. Harmark supplies a canonical form of the metric of a stationary axial symmetric manifold.         

\begin{thm}[Harmark, 2004 \cite{harmark}]
Let $(M,g)$ be Ricci-flat and let $V_i$, $i=1,...,d-1$, be $d-1$ commuting Killinfields, which fulfill the condition
\bdm
e_\rho\wedge e_{\mu_1}\wedge...\wedge e_{\mu_{d-1}}\left(\sharp\Ric(V_i),V_{\mu_1},...,V_{\mu_{d-1}}\right)=0            
\ \ \ \forall i,\rho,\mu_j=1,...,d-1,
\edm
then there exists a coordinate system $(x^1,...,x^{d-1},r,z)$, such that it holds $V_i=\frac{\partial}{\partial x^i}$ 
and in which $g$ has the form 
\bdm
g = \sum_{i,j=1}^{d-1}G_{ij}dx^i\otimes dx^j + e^{2\nu}(dr^2+dz^2).
\edm
Thereby $r=\sqrt{|det(G_{ij})|}$, $det(G_{ij})\neq const.$ and $G_{ij}=G_{ij}(r,z)$, $\nu=\nu(r,z)$. 
This form of the metric is called \emph{canonical form} or \emph{generalized Weyl-Papapetrou-Form}.
\end{thm}

\begin{rem}
\begin{itemize}
	\item[]
	\item On stationary and axial symmetric manifolds the group $\R\x\U(1)^{d-1}$ acts per definition isometrically. Because 
	      of this action $d-1$ commuting Killingfields are given.
	\item In components, the condition of the prior theorem reads
	      \bdm
	       V_i^\nu Ric_\nu^{[\rho}V_1^{\mu_1}V_2^{\mu_2}\cdot...\cdot V_{d-1}^{\mu_{d-1}]}=0\ \ \ \forall i,\rho,\mu_j=1,...,d-1.
	      \edm
	\item One can reason that solutions of the Einstein equation which are asymptotically the $4$- or $5$-dimensional Minkowskispace,
	      always satisfy the conditions of the prior theorem. For $d=3$ these conditions are always satisfied. See again \cite{harmark} 
	      for a justification of these statements.
	\item For $d=3$ and $G_{11}=-e^{2U}$, $G_{12}=-e^{2U}A$ and $G_{22}=e^{-2U}(r^2-A^2e^{4U})$ in the coordinates $x^1=t$
	      and $x^2=\phi$ one gets the well-known \emph{Papapetrou-Form} 
	      \bdm
	      g=-e^{2U}\left(dt+Ad\phi\right)^2+e^{-2U}r^2d\phi^2+e^{-2\nu}\left(dr^2+dz^2\right),
	      \edm
	      which serves as an ansatz for the Kerr metric.
\end{itemize}
\end{rem}

\section{Spherical symmetry: The Tangherlini metric}
In 1963 Tangherlini found in \cite{tangh} a generalization of the Schwarzschild metric in such a way that the dimensionality 
$d+1$ of spactime is arbitrary:
\begin{equation}
g_{T}:=-\left(1-\frac{\mu}{r^{d-2}}\right)dt^2+\frac{1}{\left(1-\frac{\mu}{r^{d-2}}\right)}dr^2+r^2 g_{S^{d-1}}\label{eq.ST},
\end{equation}
where $\mu$ describes the mass-parameter $\mu=\frac{4\pi M}{(d-1)\Omega_{d-1}}$, in which $\Omega_{d-1}$ denotes the volume 
of the $(d-1)$-dimensional unit sphere and $M$ the mass of the gravitating object in the far field. Setting $d=3$ yields the
Schwarzschild metric. We assume that $\mu$ and $r$ are strictly positive. After comparing with definition \ref{dfn.sph} we 
see that $g_{T}$ is stationary and axial symmetric for $r^{d-2}>\mu$. 
We want to call the hypersurface $\{r^{d-2}=\mu\}$ \emph{Tangherlini sphere}, which is given as the set of roots of the 
function $\Delta_T:=1-\frac{\mu}{r^{d-2}}$ and which generalizes the \emph{Schwarzschild sphere}. Because the latter carries 
the properties of an event horizon, we want to call $\Delta_T$ \emph{horizon function}.
Also, the Tangherlini metric is asymptotically flat. It is shown, \cite{birk} that the theorem of Birkhoff is independent 
of the dimension of spacetime. That means that every stationary and spherical symmetric solution of the Einstein equation in 
$(d+1)$-dimension belongs to the family of the Tangherlini metrics. A proof of $g_{T}$ actually being a solution of the 
Einstein equation is given in the appendix.

\subsection{Effective motions in Tangherlini spacetime}
Consider now a timelike geodesic $\gamma(s)=(t(s),r(s),\chi_1(s),...,\chi_{d-2}(s))$ with $r>r^{d-2}$ for all $s\in\R$. 
We use the equivalence of the geodesic equation with the Euler-Lagrange equation 
$\frac{\partial \Lie}{\partial x^i}=\frac{d}{dt}\frac{\partial \Lie}{\partial\dot{x}^i}$, with the lagrangian          
\begin{eqnarray*}
2\Lie=g_{T}(\dot{\gamma},\dot{\gamma})&=& -\left(1-\frac{\mu}{r^{d-2}}\right)\dot{t}^2+\frac{1}{\left(1-\frac{\mu}{r^{d-2}}
\right)}\dot{r}^2\\
                                    & &+r^2\left(\dot{\chi_1}^2+\sin^2\chi_1\dot{\chi_2}^2+...+\prod_{s=1}^{d-3}\sin^2\chi_s\dot{\chi_{d-2}}^2\right).
\end{eqnarray*}
The dot $\dot{}$ refers to differentiation with respect to the proper time $s$. We consider plane 
motions that means $\chi_i=\frac{\pi}{2}$ for all $i>1,\ s\in\R$. The fact that $\partial_t$ and $\partial_{\chi_1}$ are 
Killingvectors is equivalent to $t$ and $\chi_1$ being cyclic. It thus holds
\begin{eqnarray*}
-\frac{\partial \Lie}{\partial\dot{t}}        &=& \left(1-\frac{\mu}{r^{d-2}}\right)\dot{t}=const.=:E\\
 \frac{\partial \Lie}{\partial\dot{\chi_{1}}} &=& r^2\dot{\chi_1} = const. =:L.                                  
\end{eqnarray*}
If we plug these equations into $2\Lie=-1$ (we consider timelike motions), the following equation reveals:
\bdm
1=\frac{1}{\left(1-\frac{\mu}{r^{d-2}}\right)}E^2-\frac{1}{\left(1-\frac{\mu}{r^{d-2}}\right)}\dot{r}^2-\frac{L^2}{r^2}.
\edm
Tranforming this equation one gets the equation for the energy of the system
\begin{equation}\label{eq.energie}
E^2=\dot{r}^2+V(r),
\end{equation}
with effective potential
\bdm
V(r):=\left(\frac{L^2}{r^2}+1\right)\left(1-\frac{\mu}{r^{d-2}}\right).
\edm
Considering lightlike motions that means $2\Lie=0$, one gains by means of analogous calculations equation \ref{eq.energie} for the 
energy of the system with effective potential 
\bdm
\tilde{V}(r):=\frac{L^2}{r^2}\left(1-\frac{\mu}{r^{d-2}}\right).
\edm

The values of $V(r)$ converge to $-\infty$, if $r\ra0$, and to $1$, if $r\ra\infty$. We now want to find out, how this potential 
behaves in between. For the existence of extremals we have to find roots of the derivative:
\begin{eqnarray}
& & \frac{d}{dr}V(r)=-\frac{2L^2}{r^3}+\frac{L^2\mu d}{r^{d+1}}+\frac{(d-2)\mu}{r^{d-1}} = 0 \\
&\Longleftrightarrow &  r^{d-2}-\frac{(d-2)\mu}{2L^2}r^2-\frac{\mu d}{2} = 0.\label{eq.extrem}
\end{eqnarray}
For a criterion, if the extremals are local minima or maxima, we analyze the second derivative of $V(r)$:
\begin{eqnarray}
& & \frac{d^2}{dr^2}V(r)=\frac{6L^2}{r^4}-\frac{d(d+1)L^2\mu}{r^{d+2}}+\frac{(d-2)(d-1)\mu}{r^d}\gtrless 0\\
&\Longleftrightarrow& r^{d-2}-\frac{(d-2)(d-1)}{6}\frac{\mu}{L^2}r^2-\frac{d(d+1)}{6}\mu\gtrless 0.\label{eq.max}
\end{eqnarray}
We will now focus on the cases $d=3$ and $d=4$. Let's start with $d=3$. Equation $(\ref{eq.extrem})$ is solved by                  
\bdm
r = \pm L\sqrt{\frac{L^2}{\mu^2}-3}+\frac{L^2}{\mu}.
\edm
Because $L\sqrt{L^2/\mu^2-3}<L^2/\mu$, both solutions are indeed positive. And since there exist exactly two extremals 
because of the given asymptotics of the potential, they have to be one minimum and one maximum, where the minimum is taken 
at a higher value of $r$ than the maximum. We see furthermore that for $\frac{L}{\mu}<\sqrt{3}$ no extremals exist and also 
no closed orbits. In particular every particle with $E^2<1$ moves with increasing velocity onto the Schwarzschild sphere. 

Let now be $d=4$. In this case $(\ref{eq.extrem})$ is solved by
\bdm
r = \sqrt{\frac{\mu}{1-\frac{\mu}{L^2}}}.
\edm
Because we only want to consider positive values of $r$, only the positive root is of interest here. Is $d=4$, inequality 
$(\ref{eq.max})$ is equivalent to $\left(1-\frac{\mu}{L^2}\right)r^2-\frac{5}{3}\mu\gtrless 0$ and we see that 
$r=\sqrt{\frac{\mu}{1-\frac{\mu}{L^2}}}$ for $L^2\geq\mu$ is a local maximum. Is $L^2<\mu$, no extremals exist and again, 
a particle with energy $E^2<1$ would move with increasing velocity onto the Tangherlini sphere.
In particular, local minimal do not exist for whatever values of $L$ and $\mu$, what means that no \emph{stable} bounded orbits exist.
The conjecture is that only for $d=3$ there exist stable circular orbits. 
In Figure \ref{fig.potential} the function $\sqrt{V(r)}$ is pictured for different values of the angular momentum $L$ in dimensions 
$d=3,4,5$ with $\mu=1$.
\begin{figure}[hp]
 \centering
 \includegraphics[width=5cm]{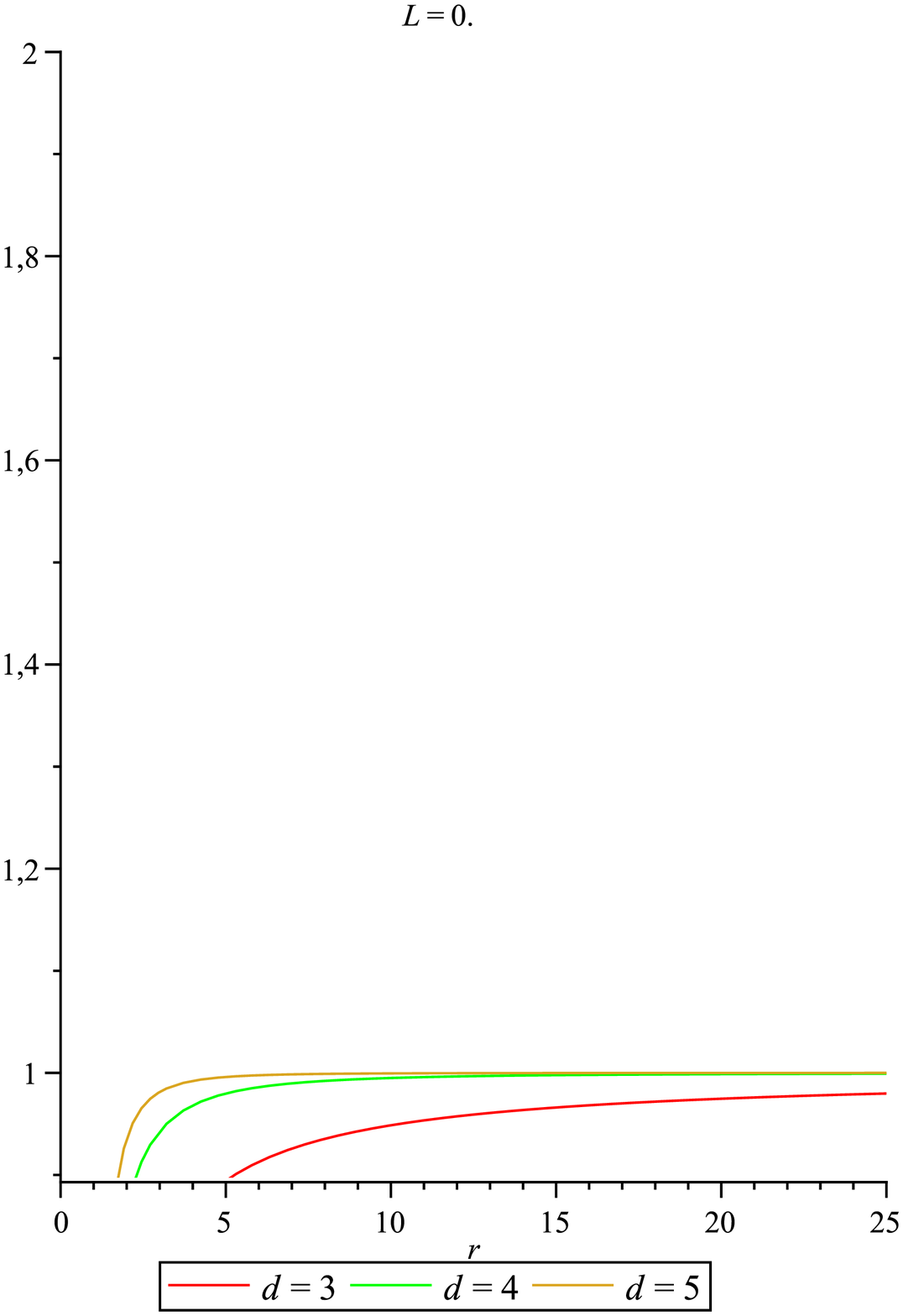}
 \includegraphics[width=5cm]{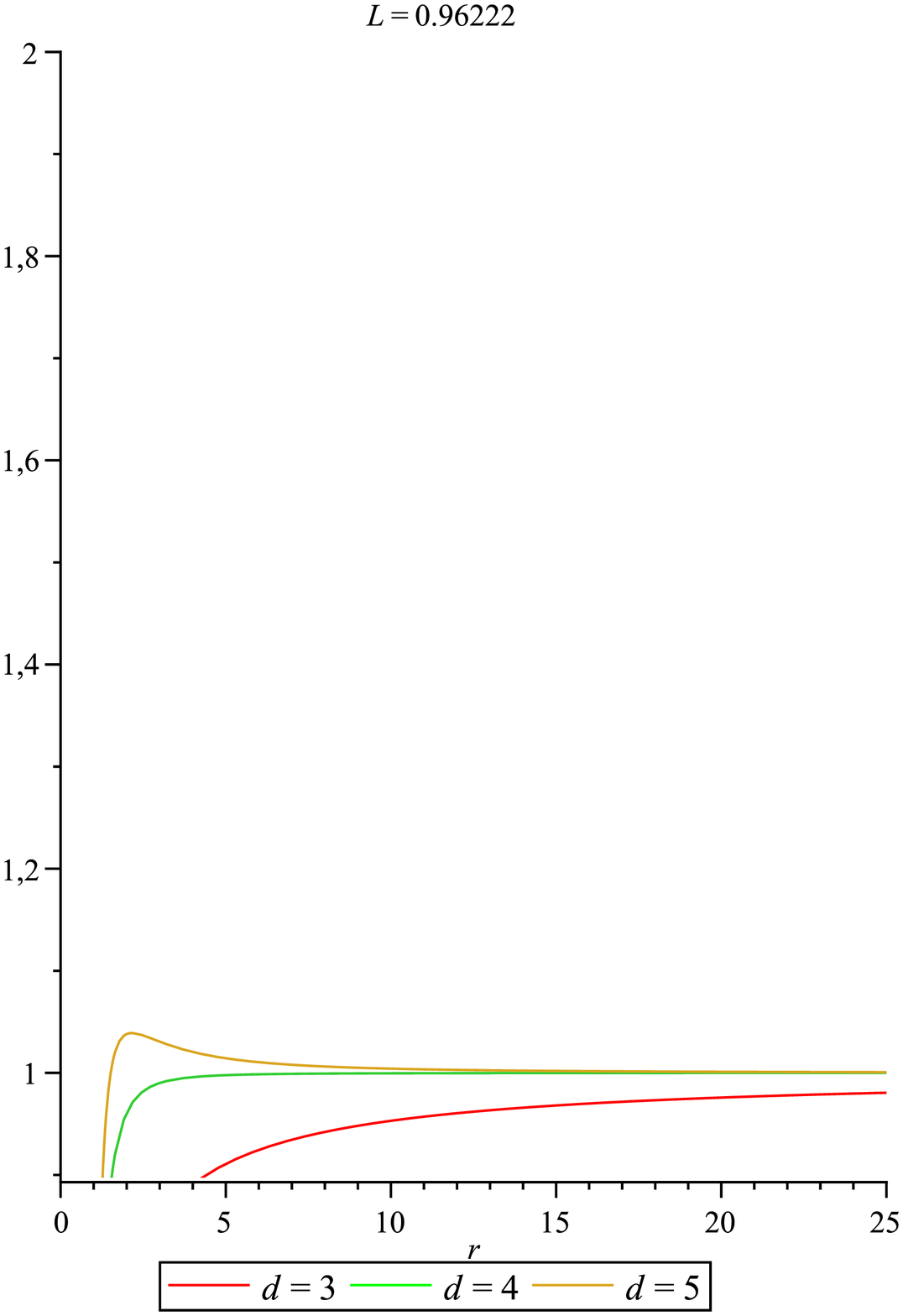}
 \includegraphics[width=5cm]{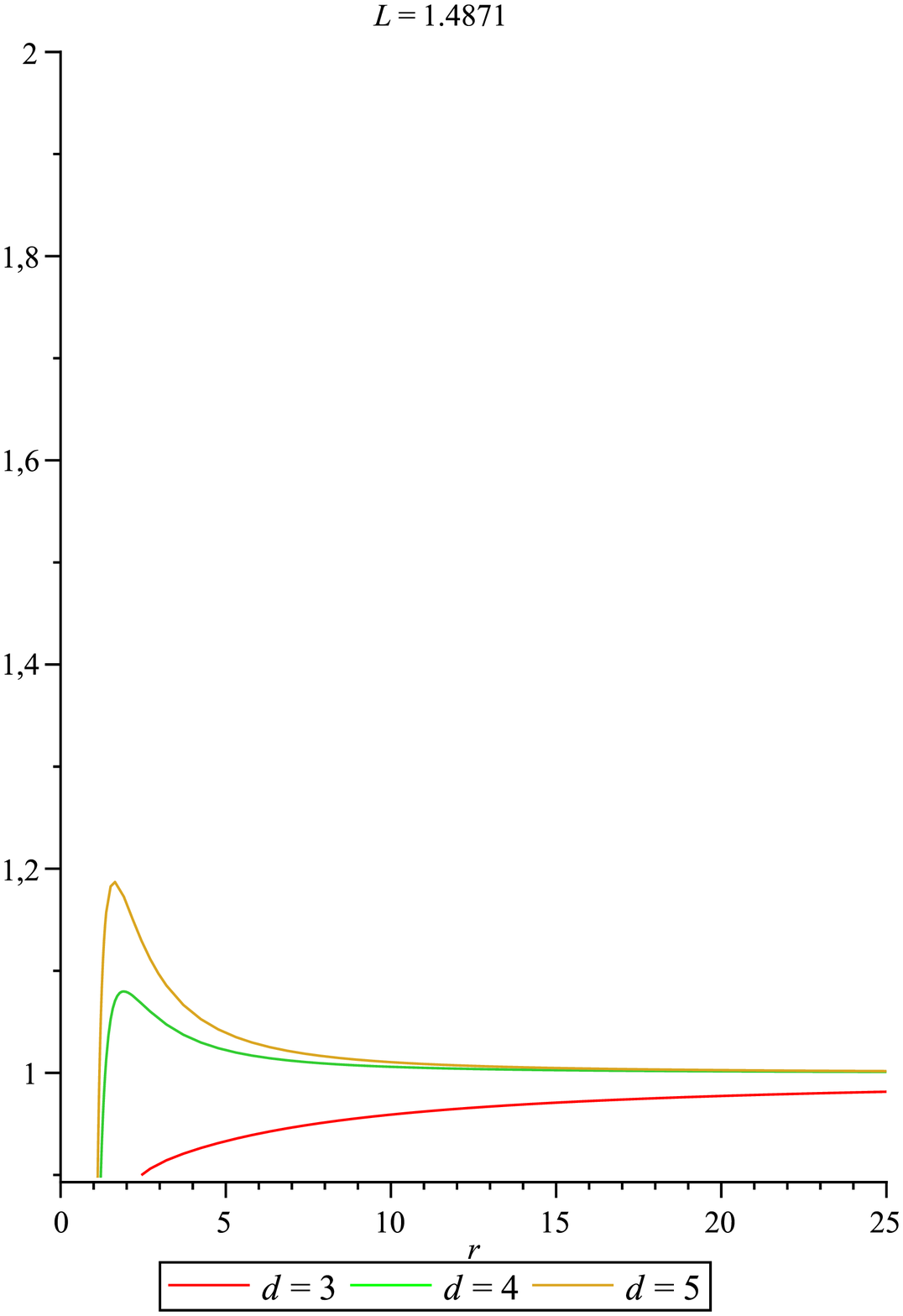}
 \includegraphics[width=5cm]{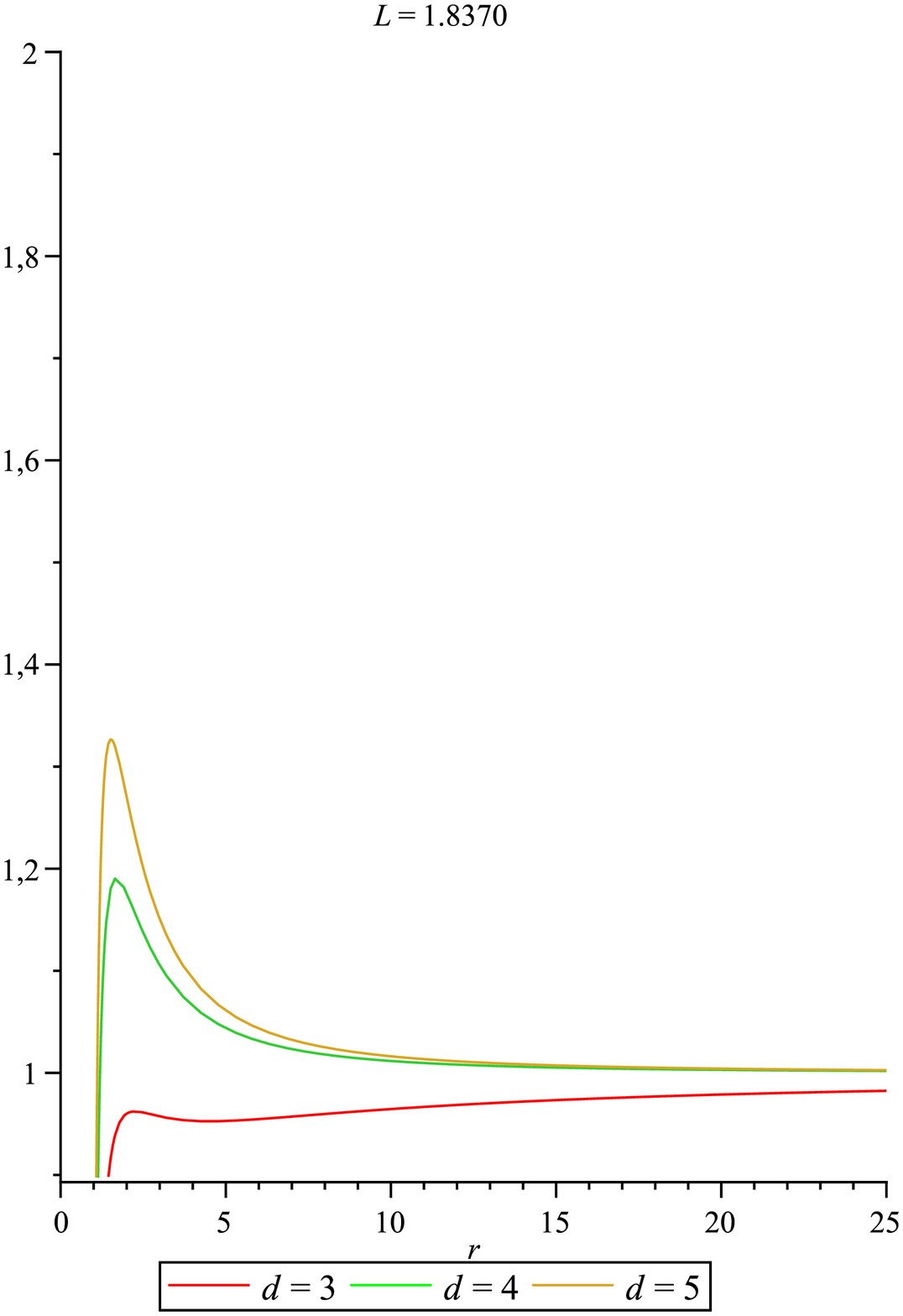}
 \includegraphics[width=5cm]{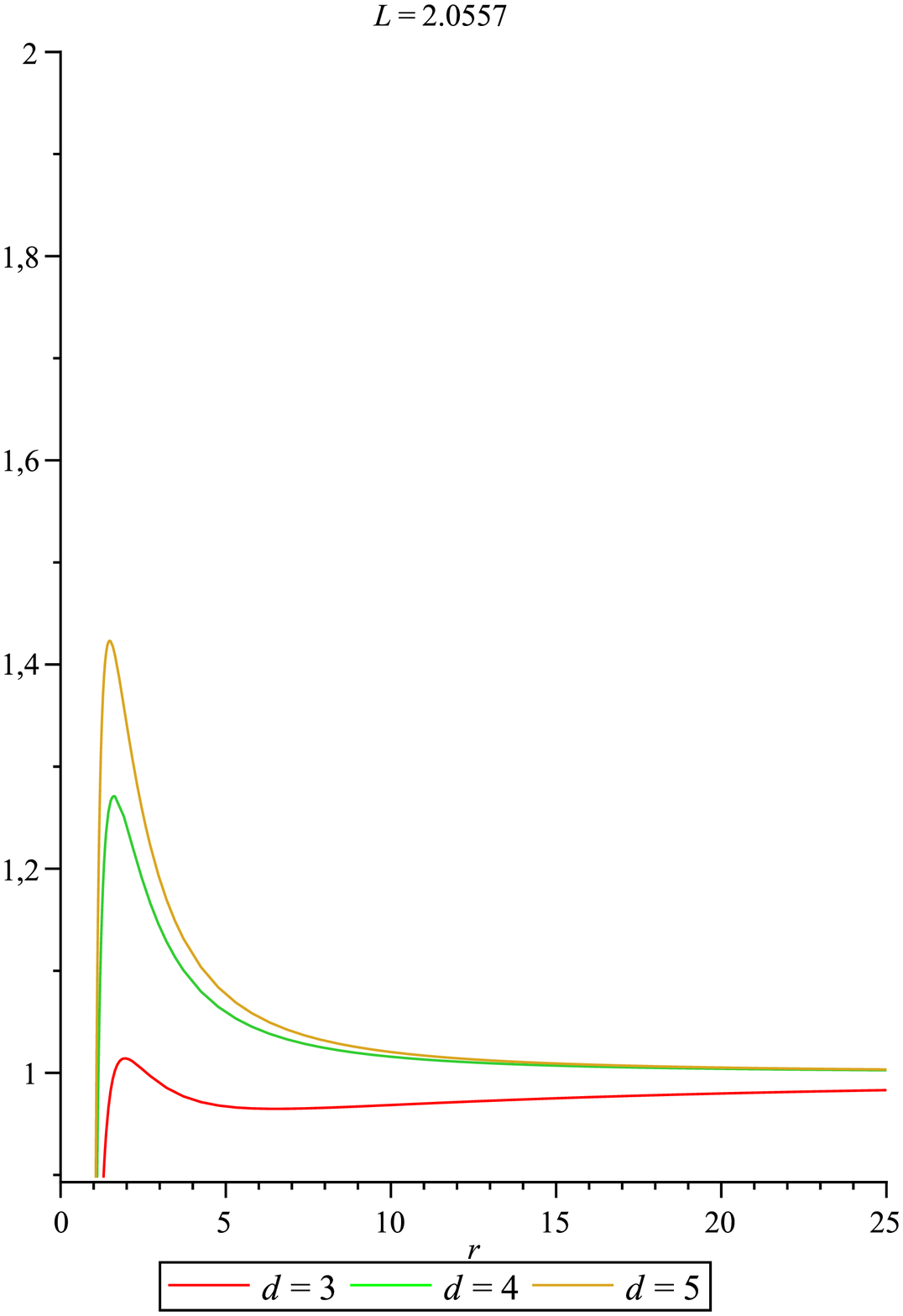}
 \includegraphics[width=5cm]{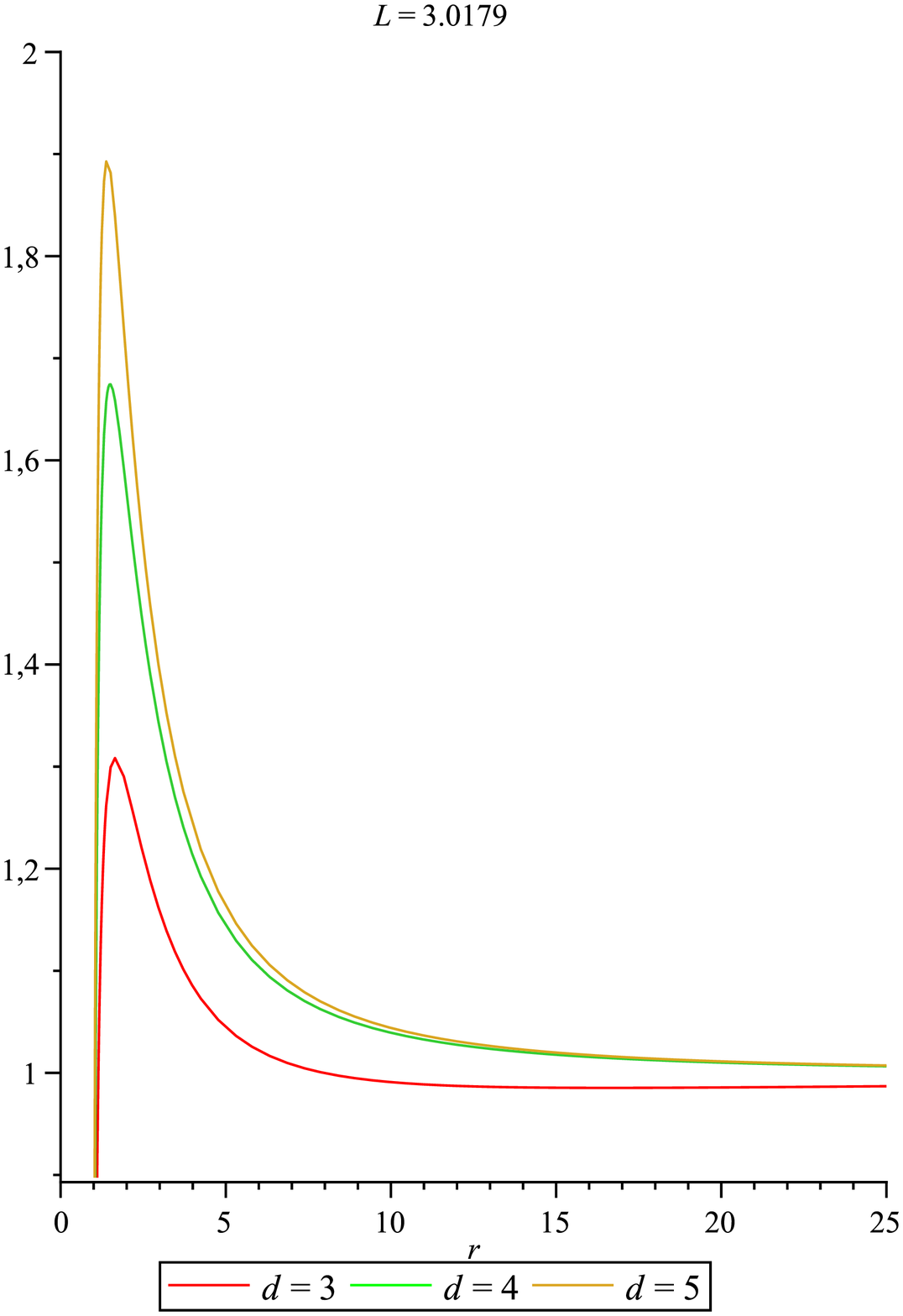}
 \caption{For $\mu=1$ the function $\sqrt{V(r)}$ is plotted, for different values of the angular momentum $L$ in the dimensions $d=3,4,5$.}
 \label{fig.potential}
\end{figure}
\subsection*{}
We now derive the differential equation for the plane motion $r(\varphi)$. First of all it holds 
$\dot{r}=r'\dot{\varphi}=r'\frac{L}{r^2}$ because of $r':=\frac{\partial r}{\partial\varphi}=\frac{\dot{r}}{\dot{\varphi}}$ 
and $L\equiv r^2\dot{\chi_1}$. Plugging this into the energy equation $\dot{r}^2+V(r)=E^2$, one obtains
\bdm
r'^2\frac{L^2}{r^4}=E^2-V(r).
\edm
Next, we perform the change of coordinates $u=1/r$. With this, it holds $r'=-\frac{u'}{u^2}$ and therefore
\begin{eqnarray*}
& & L^2u'^2=E^2-(1-\mu u^{d-2})(1+L^2u^2)= E^2-1-L^2u^2+\mu u^{d-2}+\mu L^2 u^d\\
&\Longleftrightarrow& L^2u'^2+L^2u^2= E^2-1+\mu u^{d-2}+\mu L^2 u^d\\
&\Longleftrightarrow& u'^2+u^2=\frac{E^2-1}{L^2}+\frac{\mu}{L^2}u^{d-2}+\mu u^d.\ \ \ \ \ \ \ \ \ \ \ \ \ (\star)
\end{eqnarray*}
Differentiating this expression with respect to $\varphi$, one obtains                                                      
\bdm
2u'u''+2uu'=\frac{(d-2)\mu}{L^2}u'u^{d-3}+d\mu u'u^{d-1}.
\edm
It follows that either $u'=0$, which is equivalent to $r=const.$ and therefore corresponds to circular motion, or $u$ behaves 
corresponding to the equation
\bdm
u''+u=\frac{(d-2)\mu}{2L^2}u^{d-3}+\frac{d}{2}\mu u^{d-1}.
\edm

For $d=3$ and $\mu\equiv2m$ the solution of this orbital equation is a modification of the Kepler ellipse 
$u(\varphi)=\frac{m}{L^2}(1+e\cos\varphi)$ with eccentricity $e$:
\bdm
u(\varphi)=\frac{m}{L^2}(1+e\cos\varphi)+\frac{3m^3}{L^4}\left(1+\frac{e^2}{2}-\frac{e^2}{6}\cos2\varphi+e\varphi\sin\varphi\right).
\edm

If one plugs in $u'=0$ into equation $(\star)$, one obtains the circular orbits dependence of the existence from the 
energy $E$ and the angular momentum $L$ of a testparticle at the point $u$:
\bdm
u^d+\frac{1}{L^2}u^{d-2}-u^2=\frac{1-E^2}{\mu L^2}.
\edm
Even in dimensions $d+1=4$ and $d+1=5$ the solutions are quite complicated expressions and are therefore omitted here. But in 
principle they are easy to calculate.

\begin{rem}
For another approach calculating the effective orbital potential of the Tangherlini metric see \cite{tangh}, p. 645.
\end{rem}

\subsection{The Kruskal continuation of the Tangherlini spacetime}
In this subsection we want to see that the Tangherlini metric possesses a continuation on $r^{d-2}\leq\mu$. This will be a
generalization of the known \emph{Kruskal continuation} of the Schwarzschild metric. The associated calculations generalize 
those of \cite{strau}.
At first we observe that space and time switch their role at $r^{d-2}=\mu$. Namely it holds
\bdm
g_{tt}  =  -\left(1-\frac{\mu}{r^{d-2}}\right),\quad 
g_{rr}  =  \frac{1}{1-\frac{\mu}{r^{d-2}}}.
\edm
This means for $r>\sqrt[d-2]{\mu}$,\quad $\partial_t$ is timelike and $\partial_r$ is spacelike. For $r<\sqrt[d-2]{\mu}$\quad 
however $\partial_t$ is spacelike and $\partial_r$ is timelike. Furthermore it is known that in four spacetime dimensions a 
testparticle takes infinitely long coordinate time $t$ to reach the sphere $r^{d-2}=\mu$, whereas it only needs finite proper time.
This indicates that the coordinates $t$ and $r$ are not adequate for the physical circumstances at $r=\sqrt[d-2]{\mu}$. 
Therefore we try to introduce new coordinates $(u,v)$ which are more appropriate to the geometry. We get a hint how to do this 
by looking at the description of the behaviour of the lightcones. Consider a light cone in radial direction, the Schwarzschild 
metric yields a description of this motion by
\bdm 
\frac{dr}{dt}=\pm\left(1-\frac{\mu}{r^{d-2}}\right).
\edm
If $r\downarrow \sqrt[d-2]{\mu}$, the opening angle of the light cone becomes infinitesimally small, which means that a test 
particle in this inertial system gets accelerated to the velocity of light when moving to the sphere $r=\sqrt[d-2]{\mu}$.
The following Ansatz for the metric in the new coordinates $(u,v)$ therefore seems to be appropriate:
\begin{equation}
g_{T}=-f^2(u,v)(dv\otimes dv-du\otimes du)+r^2g_{s^{d-1}}.\label{eq.0}
\end{equation}
It now holds  $(du/dv)^2=1$, for $f^2\neq0$, this means constant opening angles of the light cones for radial movements. Thus, 
we are looking for a coordinate transformation $h:(r,t)\mapsto (u,v)$ under which the Tangherlini metric behaves like
\bdm
h^*\left(-f^2(u,v)(dv\otimes dv-du\otimes du)\right)=-\left(1-\frac{\mu}{r^{d-2}}\right)dt\otimes dt+\frac{1}{1-\frac{\mu}{r^{d-2}}}dr\otimes dr,
\edm
for an $f=f(u,v)$ with $(h^*f)^2\neq0$ at $r=\sqrt[d-2]{\mu}$. In components this equation reads
\begin{eqnarray*}
 \left(1-\frac{\mu}{r^{d-2}}\right)&=& f^2\left(\left(\frac{\partial v}{\partial t}\right)^2-\left(\frac{\partial u}{\partial t}\right)^2\right),\\
-\left(1-\frac{\mu}{r^{d-2}}\right)&=& f^2\left(\left(\frac{\partial v}{\partial r}\right)^2-\left(\frac{\partial u}{\partial r}\right)^2\right),\\
           0 &=& \frac{\partial u}{\partial t}\cdot\frac{\partial u}{\partial r}-\frac{\partial v}{\partial t}\cdot\frac{\partial v}{\partial r}.
\end{eqnarray*}
To simplify calculations, we introduce a new radial coordinate $r^*:=r+\mu \ln{\left(\frac{r^{d-2}}{\mu}-1\right)}$ and a 
function $F(r^*):=\frac{1}{\tilde{f}^2(r)}\left(1-\frac{\mu}{r^{d-2}}\right)$, where $\tilde{f}:=h^*f$. 

We assumed that it is possible to find a coordinate transformation which behaves like $h^*f=h^*f(r)$.
With this, the above equations take the following form:
\begin{eqnarray}                                                                                             
 F(r^*) &=& \left(\frac{\partial v}{\partial t}\right)^2-\left(\frac{\partial u}{\partial t}\right)^2,\label{eq.1}\\                  
-F(r^*) &=& \left(\frac{\partial v}{\partial r^*}\right)^2-\left(\frac{\partial u}{\partial r^*}\right)^2,\label{eq.2}\\ 
\frac{\partial u}{\partial t}\cdot\frac{\partial u}{\partial r^*} &=& \frac{\partial v}{\partial t}\cdot\frac{\partial v}{\partial r^*}.\label{eq.3}
\end{eqnarray}
Taking skillfull linear combinations, namely $(\ref{eq.1})+(\ref{eq.2})\pm 2\cdot (\ref{eq.3})$, we obtain
\begin{eqnarray*}
\left(\frac{\partial v}{\partial t}+\frac{\partial v}{\partial r^*}\right)^2 &=& \left(\frac{\partial u}{\partial t}+\frac{\partial u}{\partial
                                                                                  r^*}\right)^2,\\
\left(\frac{\partial v}{\partial t}-\frac{\partial v}{\partial r^*}\right)^2 &=& \left(\frac{\partial u}{\partial t}-\frac{\partial u}{\partial
                                                                                  r^*}\right)^2.
\end{eqnarray*}
Taking the square root out of both equations and choosing the positive sign of the root for the first equation and the negative 
sign for the second equation leads to the result that the Jacobi-Determinant doesn't vanish. We now get
\bdm
\frac{\partial v}{\partial t}=\frac{\partial u}{\partial r^*},\ \ \ \frac{\partial v}{\partial r^*}=\frac{\partial u}{\partial t}.
\edm
Differentiating the first equation with respect to $r^*$ and the second equation with respect to $t$, one can deduce the 
following wave equation:
\bdm
\frac{\partial^2 u}{\partial t^2}-\frac{\partial^2 u}{\partial r^{*2}}=0,\ \ \ \frac{\partial^2 v}{\partial t^2}-\frac{\partial^2 v}{\partial r^{*2}}=0.
\edm
The most general solution is
\begin{eqnarray*}
v &=& h(r^*+t)+g(r^*-t)\\
u &=& h(r^*+t)-g(r^*-t).
\end{eqnarray*}
We now plug in these expressions for $u$ and $v$ into the equations $(\ref{eq.1})$ to $(\ref{eq.3})$. At first we discover 
that equation $(\ref{eq.3})$ is fullfilled identically and thus leads to no new condition. Equations $(\ref{eq.1})$ and $(\ref{eq.2})$
on the other hand provide the condition $F(r^*)=(h'-g')^2-(h'+g')^2=-(h'+g')^2+(h'-g')^2$, which leads to the following 
identity for $F(r^*)$:
\bdm
F(r^*)=-4 h'(r^*+t)g'(r^*-t).
\edm
Differentiating this expression, once with respect to $r^*$ and once with respect to $t$, we get
\begin{eqnarray}
F'(r^*) &=& -4(h''g'+h'g'')\label{eq.4}\\
      0 &=& -4(h''g'-h'g'').\label{eq.5}
\end{eqnarray}
For now, we assume $r>\sqrt[d-2]{\mu}$. In this case is holds $F(r^*)>0$ and from $(\ref{eq.4})$ and $(\ref{eq.5})$
we can deduce the equations
\begin{eqnarray*}
\frac{F'(r^*)}{F(r^*)} &=& \frac{h''(r^*+t)}{h'(r^*+t)}+\frac{g''(r^*-t)}{g'(r^*-t)}\\
                     0 &=& \frac{h''(r^*+t)}{h'(r^*+t)}-\frac{g''(r^*-t)}{g'(r^*-t)}.
\end{eqnarray*}
And with this
\bdm
\frac{F'(r^*)}{F(r^*)} = 2\frac{h''(r^*+t)}{h'(r^*+t)},
\edm
which is equivalent to 
\begin{equation}
(\ln{F(r^*)})'= 2(\ln{h'})'(r^*+t)\label{eq.6}.
\end{equation}
In this formula, both sides have to be equal to the same constant, which we will call $2\eta$. With the choice of the integration
constant $c$ for the left hand side as $c=\ln{\eta^2}$, it follows $\ln F(r^*)=2\eta r^*+\ln\eta^2$, what means that 
$F(r^*)=\eta^2\exp(2\eta r^*)$. 
Defining $y:=r^*+t$ and choosing $\ln\frac{\eta}{2}$ as the integration constant of the right hand side that means 
$\ln h'=\eta y+\ln\frac{\eta}{2}$, it holds $h'=\frac{\eta}{2}\exp(\eta y)$ and therefore, $h=\frac{1}{2}\exp(\eta y)$.

By means of formula $(\ref{eq.5})$ it is now also possible to find an expression for $g(y)$. Namely because of 
$h''=(\frac{\eta^2}{2}e^{\eta y})$ it holds
\bdm
0=\frac{\frac{\eta^2}{2}e^{\eta y}}{\frac{\eta}{2}e^{\eta y}}-\frac{g''(y)}{g'(y)} \Leftrightarrow
\eta=\frac{g''(y)}{g'(y)} \Leftrightarrow
g''(y)=\eta g'(y)         \Leftrightarrow
g'(y)=C e^{\eta y}.
\edm
Choosing $C=-\frac{\eta}{2}$ we obtain the follwing expressions:
\begin{gather}
g(y)=-\frac{1}{2}e^{\eta y},\ h(y)=\frac{1}{2}\exp(\eta y),\ F(r^*)=\eta^2\exp(2\eta r^*).\label{eq.7}
\end{gather}
With this we can now determine $u$ and $v$:
\bdm
 u=h(r^*+t)-g(r^*-t)=\frac{1}{2}e^{\eta}(r^*+t)+\frac{1}{2}e^{\eta(r^*-t)}=e^{\eta r^*}\cosh(\eta t)
  =e^{\eta\left(r+\mu\ln\left(\frac{r^{d-2}}{\mu}-1\right)\right)}\cosh(\eta t),
\edm
that means
\bdm
 u=e^{\eta r}\left(\frac{r^{d-2}}{\mu}-1\right)^{\mu\eta}\cosh(\eta t).
\edm
Analogously it holds 
\begin{eqnarray*}
 v &=& h(r^*+t)+g(r^*-t)=\frac{1}{2}e^{\eta(r^*+t)}-\frac{1}{2}e^{\eta(r^*-t)}=e^{\eta r^*}\sinh(\eta t)\\
   &=& e^{\eta(r+\mu\ln\left(r+\mu\ln\left(\frac{r^{d-2}}{\mu}-1\right)\right))}\sinh(\eta t),
\end{eqnarray*}
and thus
\bdm
 v=e^{\eta r}\left(\frac{r^{d-2}}{\mu}-1\right)^{\mu\eta}\sinh(\eta t).
\edm
Furthermore, with the expression for $F(r^*)$ from $(\ref{eq.7})$ it holds
\begin{eqnarray*}
\tilde{f}^2 &=& 
              \frac{1-\frac{\mu}{r^{d-2}}}{F(r^*)}=\frac{1-\frac{\mu}{r^{d-2}}}{\eta^2e^{2\eta\left(r+\mu\ln\left(\frac{r^{d-2}}{\mu}-1\right)\right)}}
           =  \frac{1-\frac{\mu}{r^{d-2}}}{\eta^2e^{2\eta r}\left(\frac{r^{d-2}}{\mu}-1\right)^{2\mu\eta}}\\
          &=& \frac{1}{\eta^2}e^{-2\eta r}\frac{\mu}{r^{d-2}}\frac{\frac{r^{d-2}}{\mu}-1}{\left(\frac{r^{d-2}}{\mu}-1\right)^{2\mu\eta}},
\end{eqnarray*}
and with this
\bdm
\tilde{f}^2=\frac{\mu}{\eta^2r^{d-2}}e^{-2\eta r}\left(\frac{r^{d-2}}{\mu}-1\right)^{1-2\mu\eta}.
\edm
Now, $\eta$ is chosen so that $\tilde{f}^2\neq 0$ at $r=\sqrt[d-2]{\mu}$. That means $\eta=\frac{1}{2\mu}$. It follows
\bdm
\tilde{f}^2=\frac{\mu}{r^{d-2}(2\mu)^2e^{-\frac{r}{\mu}}}=\frac{4\mu^3}{r^{d-2}}e^{-\frac{r}{\mu}}.
\edm
The generalized Kruskal transformation is thus given by
\begin{eqnarray*}
 u &=& \sqrt{\left(\frac{r^{d-2}}{\mu}-1\right)}e^{r/2\mu}\cosh\left(\frac{t}{2\mu}\right)\\
 v &=& \sqrt{\left(\frac{r^{d-2}}{\mu}-1\right)}e^{r/2\mu}\sinh\left(\frac{t}{2\mu}\right).
\end{eqnarray*}
and in these coordinates the Tangherlini metric has the form $(\ref{eq.0})$ with 
\bdm
\tilde{f}^2=\frac{4\mu^3}{r^{d-2}}e^{-r/\mu}.
\edm

To derive equation $(\ref{eq.6})$, we made the assumption that $r>\sqrt[d-2]{\mu}$. Hence we only found a coordinate 
transformation so far, but no continuation for $r\leq\sqrt[d-2]{\mu}$. However, $\tilde{f}^2$ is also defined for 
$0<r\leq\sqrt[d-2]{\mu}$. To see this, we first consider the following equations:
\begin{eqnarray}
u^2-v^2     &=& \left(\frac{r^{d-3}}{\mu}-1\right)e^r/\mu\label{eq.hyp}\\
\frac{v}{u} &=& \tanh{\frac{t}{2\mu}}\label{eq.line}.
\end{eqnarray}
The region $r>\sqrt[d-2]{\mu}$\quad corresponds to the region $u>|v|$. And equation $(\ref{eq.hyp})$ says that those points 
of the $(t,r)$-plane, where $r=const.$ correspond to hyperbolas of the $(u,v)$-plane (see figure \ref{fig.kruskal} on page 
\pageref{fig.kruskal}). For $r\ra\sqrt[d-2]{\mu}$, the hyperbolas cling more and more to the bisecting lines, because in this 
limit it holds $u^2=v^2$. Because of $(\ref{eq.line})$ the lines $t=const.$ correspond to the lines through the origin. For 
$t\rightarrow\pm\infty$ it holds $\tanh\rightarrow\pm1$. This limit coincides with $\{r=2m\}$. 
The metric is not defined on the hyperbolas $v^2-u^2=1$, because these points correspond to $r=0$. Is however $v^2-u^2<1$, 
that is $0<r\leq\sqrt[d-2]{\mu}$, then the right hand side of $(\ref{eq.hyp})$ is monotonely increasing and therefore $r$ is 
a well-defined function of $u$ and $v$. This is why $\tilde{f}^2$ cannot be singular at these points. 

The essence of the Kruskal transformation is therefore in particular qualitatively the same in every dimension.
Altogether one can say that the assumption of a static spherical symmetry is very strong and restrictive which is why 
we couldn't observe dimension-dependent behaviour in the hole section. This is very different from axial symmetry.
\begin{figure}[h]
\includegraphics[width=5cm,height=5cm]{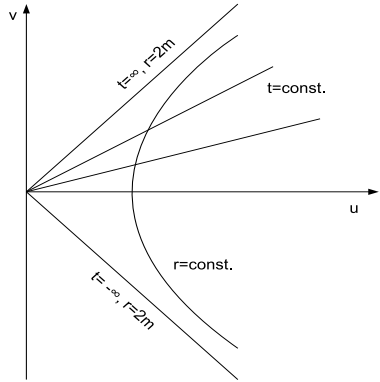}
\caption{The niveaulines $t=const.$ and $r=const.$ in the Kruskal plane.}
 \label{fig.kruskal}
\end{figure}
\section{Axial symmetry: The Myers-Perry metric}
%
\subsection{The Kerr metric}
Before considering the Myers-Perry metric, which is a family of higherdimensional axial symmetric solutions of the vacuum 
equation, let us have a brief look at its fourdimensional counterpart, the Kerr metric. From it, we want to see some crucial 
features, which lay a fundament of what features are to be watched out for in the higherdimensional case.

The Kerr metric can be interpreted as a dynamic generalization of the Schwarzschild metric. As such, it is a good 
model for the gravitational field of a rotating central-symmetric mass distribution. The Kerr metric helps thus realizing 
how spacetime changes due to rotation of mass.  This is an interesting fact, because in the Newtonian view of the world there 
is no distiction of rotating and non-rotating mass distributions.

In Boyer-Lindquist coordinates $(t,r,\vartheta,\varphi)$ (spherical coordinates) the Kerr metric is of the shape
\bdm
g_K = - dt\otimes dt+\frac{2mr}{\rho^2}\left(dt-a\sin^2\vartheta d\varphi\right)^2+\frac{\rho^2}{\Delta_K}dr\otimes dr+\rho^2 
        d\vartheta\otimes\vartheta+(r^2+a^2)\sin^2\vartheta d\varphi\otimes\varphi.
\edm
The functions $\rho$ and $\Delta_K$ are declared in the following way:
\begin{eqnarray*}
\rho^2 &=& r^2+a^2\cos^2{\vartheta},\\
\Delta_K &=& r^2-2mr+a^2.
\end{eqnarray*}
The parameter $m$ can again be interpreted as mass of the gravitating object and again we assume $m>0$ to avoid naked        
singularities. Other than the Schwarzschild metric, the Kerr metric is described by a second parameter, $a$, which can 
be interpreted as angular momentum per mass unit. Setting $a=0$ one obtains the Schwarschild metric. Analogously to the         
theorem of Birkhoff, one can show that the Kerr metric is the unique stationary and axial symmetric solution to the vacuum
equation \cite{heusler}. We will see that, in contrary to spherical symmetry, this feature is no longer valid for 
higher dimensional axial symmetric solutions.

Depending on the parameter $m$ and $a$, one distinguishes three different classes of the Kerr spacetime: 
\begin{eqnarray*}
0<a^2<m^2 & & \text{slowly rotating Kerr spacetime}\\
a^2=m^2   & & \text{extreme Kerr spacetime}\\
m^2<a^2   & & \text{fast rotating Kerr spacetime.}
\end{eqnarray*}

At $\rho^2=0$, $\Delta_K=0$ the Kerr metric is not defined, but it can be shown that the latter is a coordinate singularity. 
Similar to the Tangherlini case, we call the connected components of the point set $\{\Delta_K=0\}$ \emph{horizons}, wherefore 
the function $\Delta_K$ is again called \emph{horizon function}. Analyzing the horizon function, one can see that every class possesses 
a different horizon-structure, for it holds $\Delta_K(r)=0\Leftrightarrow r=r_{\pm}:=m\pm\sqrt{m^2-a^2}$. In the
\begin{itemize}
 \item slowly rotating Kerr spacetime $\Delta_K$ has two positive roots.
 \item extreme Kerr spacetime $r=m$ is a double root of $\Delta_K$.
 \item fast rotating Kerr spacetime $\Delta_K$ possesses no real roots.
\end{itemize}

Other than in the Schwarzschild case, the point set $\{t,r=0,\vartheta,\varphi\}$ only consists of singularities 
if $\vartheta=\frac{\pi}{2}$, because $\rho^2=0\Leftrightarrow \left(r=0\text{ und } \cos\vartheta=0\right)$. We denote this 
singularity by $\Sigma$. We can conclude that $\Sigma=\R(t)\times S^1$ where $S^1$ is the equator of the sphere at $r=0$. 
For this reason, $\Sigma$ is called \emph{ring singularity}. One can show \cite{oneill} that this is a curvature singularity.
Taking $\frac{\pi}{2}$ out of the domainon $\vartheta$, we can assume $r\in\R$.

In this article, we only want to consider the slowly rotating Kerr spacetime. The other two types are contained as 
special cases.
It is practical, to divide the set $\R^2\times S^2-\Sigma$ into so called \emph{Boyer-Lindquist blocks} I, II and III, 
which are defined in the following way by the value of $r$:
\begin{eqnarray*}
	\text{I}   &:&  r>r_+ \\
	\text{II}  &:&  r_-<r<r_+\\
	\text{III} &:&  r<r_-
\end{eqnarray*}

A further interesting aspect is the causality structure of the coordinate vectorfields on the Boyer-Lindquist blocks, which 
will be briefly summarized in the following. Because of $\rho^2>0$ and $\Delta_K>0$ on I and III, but $\Delta_K<0$ on II, it holds
(compare figure \ref{fig.KerrKausal}):\\

\begin{itemize}
  \item $\partial_r$ is spacelike on I and III, timelike on II.
	\item $\partial_\vartheta$ is spacelike everywhere.
	\item $\partial_\varphi$ is spacelike, if $r>0$ that means in any case on the blocks I and II, but also if $r\ll -1$.
	      Because then $r^2+a^2>\frac{2m|r|a^2\sin^2\vartheta}{r^2+a^2\cos^2\vartheta}$. That means $\partial_\varphi$ is 
	      spacelike only in some (negative) distance to the ringsingularity.
	\item $\partial_t$ is spacelike on II, because $g_{tt}>0\Leftrightarrow a^2\cos^2\theta<2mr-r^2$, which is fullfilled on the
	      open interval $(r_-,r_+)$, because of $2mr-r^2([r_-,r_+])=(a^2,m^2]$. Likeweise one realizes that $\partial_t$ is 
	      timelike for $r>2m$ and $r<0$.\\
\end{itemize}
For $r$ big enough that is $r>2m$, then the Boyer-Lindquist coordinates can be classically interpreted as time, distance from 
the rotating object, latitude and longitude. On block II however, $\partial_t$ and $\partial_r$ exchange their role, for 
$\partial_r$ now measures temporal and $\partial_t$ measures spatial distances, analogous to the situation in the interior of 
the Schwarzschild sphere.                                                                                                   
On block III for $r\ll-1$ the coordinates behave classically again, with the difference that now $-r$ measures the distance to 
the rotating massdistribution.
While $\partial_\vartheta$ and $\partial_r$ have constant causal character on each block, $\partial_\varphi$ and $\partial_t$ 
don't behave that clearly arranged. Those regions in the blocks I and II, on which $\partial_t$ is spacelike, are in each case 
called \emph{ergosphere}. In these regions interesting physical effects can be observed, which we won't deepen here.
\begin{figure}[h]
 \centering
 \includegraphics[width=7cm]{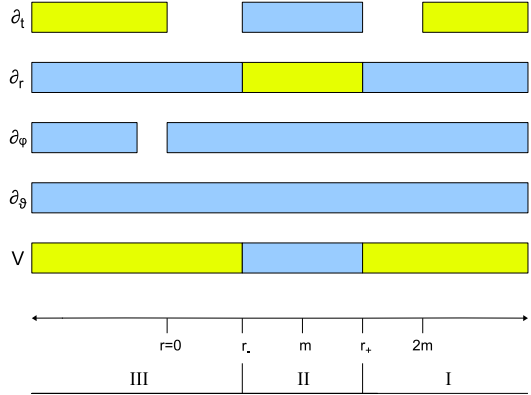}
 \caption{ The causal behaviour of the coordinate vectorfields at a glance. A yellow bar indicates timelike, a blue bar 
           spacelike behaviour.}
 \label{fig.KerrKausal}
\end{figure}
\subsection{The Myers-Perry metric}

The first property of higher dimensional axial symmetry is that there is no unique stationary solution like we have seen 
in the fourdimensional case. As an example of an axial symmetric solution we want to consider the Myers-Perry metric, which 
can be seen as a direct generalization of the Kerr metric. Other than the spherical symmetry, which is very restrictive and 
thus doesn't admit qualitatively new solutions in higher dimensions, we will discover a highly dimension-dependend behaviour 
of the Myers-Perry metric. Essential influence on the metric of a $d+1$ dimensional axial symmetric spacetime comes from the
$\left\lfloor\frac{d}{2}\right\rfloor$ possible rotationplanes, to each one can associate an angular 
momentum $J_i$. To make the qualitative behaviour of the solution more understandible, we proceed like \cite{myersperry} 
and perform the generalization in two steps and begin with rotation in just one plane.

\subsubsection{Rotation in one plane}
Considering rotation in just one plane, the Myers-Perry metric is of the shape
\begin{eqnarray*}
g_{MP1} &=& -dt\otimes dt+\frac{\mu}{r^{d-4}\rho^2}(dt-a\sin^2d\varphi)^2+\frac{\rho^2}{\Delta_{MP1}}dr\otimes dr+ \rho d\vartheta\otimes d\vartheta\\
        & & +(r^2+a^2)\sin^2\vartheta d\varphi^2+r^2\cos^2\vartheta g_{S^{d-3}},
\end{eqnarray*}
where the functions $\rho$ and $\Delta_{MP1}$ are declared analogous to the Kerr metric as 
\begin{gather*}
\rho^2=r^2+a^2\cos^2\vartheta,\ \ \ \Delta_{MP1}=r^2+a^2-\frac{\mu}{r^{d-4}}.
\end{gather*}
Comparison with the far field gives the integration constants $\mu$ and $a$ as \emph{mass-parameter} and \emph{angular momentum
per mass unit} respectively,
\begin{gather*}
\mu=\frac{4\pi m}{(d-1)\Omega_{d-1}}, \ \ \ a=\frac{J(d-1)}{2m}.
\end{gather*}
We will again assume $\mu$ to be positive. One realizes at once that for $d=3$ one obtains the Kerr metric. "Stopping" 
rotation, i.e. setting $a=0$, it yields the Tangherlini metric.

$g_{MP1}$ is singular on the sets $\{\Delta_{MP1}=0\}$ and $\{r^{d-4}\rho^2=0\}$. Because the first set is a purely coordinate 
singularity, we again call $\Delta_{MP1}$ the horizon function. Section \ref{subsec.horizon} will give a comparison of the different 
horizon functions that appear in this article. In contrary to that, the second set is a curvature singularity \cite{myersperry}.
To study the structure of the singularities, it is convenient to distinguish between $d=3$, $d=4$ and $d\geq5$ (compare table 
\ref{tab.sing}). We have already studied the case $d=3$ in the previous section.

If $d=4$, the requirement of the set $\{r^{d-4}\rho^2=0\}$ reduces to  $\rho^2=0$ and gives a ringsingularity at $r=0$ similar 
to the Kerr case. Because of this, $r$ is again defined on $\R$. The equation $\Delta_{MP1}=r^2+a^2-\mu=0$ can be solved easily by 
$r=\pm\sqrt{\mu-a^2}$ and there exist thus two horizons, if $a^2\neq\mu$. Obviously, real solutions only exist for values of $a^2$, 
which are smaller than $\mu$. In the extreme case $a^2=\mu$ the ringsingularity lies within the horizon. Is the value of 
$a^2>\mu$, then there is a naked singularity present. For a horizon to exist, the angular momentum is thus not allowed to take 
an arbitrary high value.

If $d\geq5$, the metric is singular at all points, whose $r$-coordinate is zero. This corresponds to a (in time moving) 
$(d-1)$-sphere. To get the position of the horizons, an equation of the form $r^2+a^2-\frac{\mu}{r^k}=0$ for a $k>0$ is to 
be solved, wich is equivalent to $r^{2+k}+a^2r^k=\mu$. This equation has a unique solution for $r>0$, for the function on the
left hand side is continuous and monotoneously increasing and it has the value zero for $r=0$. In particular, the existence of 
a solution is independent of the value of $a$; therefore there are also horizons for arbitrary large $a$ (which is different from
the spatial dimensions $3$ and $4$).

It appears that the dimensions $d+1=4$ and $d+1=5$ are somehow special in the Myers-Perry spacetime. But as we will realize in the next 
subsection, this feature just reflects the number of rotation planes. For $d\geq5$, one rotation plane is too little to 
cause interesting behaviour of the black hole. 

\begin{table}[h]
 \begin{tabular}{p{0.07\textwidth}|p{0.35\textwidth}|p{0.15\textwidth}|p{0.2\textwidth}|p{0.1\textwidth}}
           & \textbf{Number of horizons} & \textbf{Restriction to angular momentum} & \textbf{Type of the curvature singularity}
           &  \textbf{Domain of $r$}\\
    \hline\\
    $d=3$  & $1-2$, for $r=m\pm\sqrt{m^2-a^2}$   & $a^2\leq\frac{1}{4}\mu^2$ & Ring singularity & $r\in\R$ \\
    $d=4$  & $1-2$, for $r=\pm\sqrt{\mu-a^2}$    & $a^2\leq\mu$              & Ring singularity & $r\in\R$\\
    $d\geq5$ & $1$, for $r^{2+k}+a^2r^k=\mu$     & $a\in\R$                  & Point singularity& $r\in\R^+$\\
    \hline
 \end{tabular}
 \caption{Tabular overview of the characteristics of the different dimensions in the Myers-Perry metric.}
 \label{tab.sing}
\end{table}
It is also interesting to look at the causal character of the coordinate vectorfields, which is what we want to do now 
(compare also figures \ref{fig.MP4Kausal} and \ref{fig.MPdKausal}). For this let $d>3$.

First we analyze $\partial_r$. It holds 
\begin{eqnarray}
g_{rr}=\frac{\rho^2}{\Delta_{MP1}}\gtrless0\Longleftrightarrow \Delta_{MP1}=r^2+a^2-\frac{\mu}{r^{d-4}}\gtrless0.
\end{eqnarray}
For $d=4$ this condition is fullfilled, iff $r^2\gtrless\mu-a^2$. Outside the horizons, i.e. for $r>\sqrt{\mu-a^2}$ and 
$r<-\sqrt{\mu-a^2}$, $\partial_r$ is thus spacelike, within the horizons, $\partial_r$ is timelike.

Because for $d>4$ the $r$-component is positive, $(1)$ is equivalent to $r^{2+k}+ar^k\gtrless\mu$, if $k=d-4$. For $d>4$ 
the causal behaviour of $\partial_r$ is thus analogue to that for $d=4$.

Consider now $\partial_{\varphi}$. It holds
\bdm
g_{\varphi\varphi}=\left(r^2+a^2+\frac{\mu a^2\sin^2\vartheta}{r^{d-4}\rho^2}\right)\sin^2\vartheta>0
\edm
for all values of $r$, $\vartheta$ and $d$. This means that $\partial_{\varphi}$ is always spacelike.

Next we consider $\partial_t$. It holds
\bdm
g_{tt}=-\left(1-\frac{\mu}{r^{d-4}\rho^2}\right)\gtrless0\Longleftrightarrow \frac{\mu}{r^{d-4}\rho^2}\gtrless 1
\Longleftrightarrow \mu\gtrless r^{d-4}\rho^2,
\edm
because for $d\neq4$ always $r>0$. For $d=4$ the requirement reduces to $r^2\lessgtr\mu-a^2\cos^2\vartheta$. $\partial_t$
is thus timelike if $r>\sqrt{\mu-a^2\cos^2\vartheta}$ or if $r<-\sqrt{\mu-a^2\cos^2\vartheta}$,                        
that means in any case for $r>\sqrt{\mu}$ and $r<-\sqrt{\mu}$. $\partial_t$ is spacelike, if $r<\sqrt{\mu-a^2\cos^2{\vartheta}}$ 
and $r>0$, or $r>-\sqrt{\mu-a^2\cos^2{\vartheta}}$ and $r<0$ that means in any case for $r\in(-\sqrt{\mu-a^2},0)$ and 
$r\in(0,\sqrt{\mu-a^2})$. In the areas $\sqrt{\mu-a^2}< r<\sqrt{\mu}$ and $-\sqrt{\mu}<r<-\sqrt{\mu-a^2}$ the causal character of 
$\partial_t$ depends on $\vartheta$, similar to the ergosphere in the Kerr spacetime.

For $d>4$ there exists a number $k>0$, such that above condition can be refomulated as $\mu\gtrless r^k\left(r^2+a^2\cos^2\vartheta\right)$.
For values of $r$, for which $\mu>r^{2+k}+a^2r^k$ holds, i.e. within the horizon, $\partial_t$ is spacelike. For values of $r$, for 
which $\mu<r^{2+k}$ holds, $\partial_t$ is timelike. Within the area $r^{2+k}<\mu<r^{2+k}+a^2r^k$ the causal character again depends 
on the angle $\vartheta$.

The remaining coordinate vector fields are spacelike everywhere.

\begin{figure}[h]
 \centering
 \includegraphics[width=7cm]{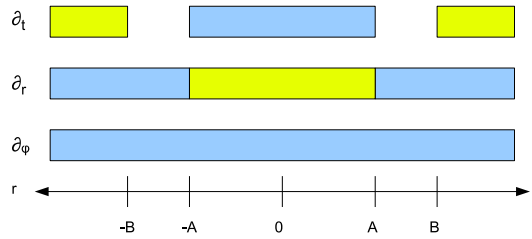}
 \caption{The causal behaviour of the coordinate vector fields of the fourdimensional Myers-Perry metric. A yellow bar 
          indicates timelike, a blue bar indicates spacelike behaviour. Here, $A:=\sqrt{\mu-a^2}$ und $B:=\sqrt{\mu}$.}
 \label{fig.MP4Kausal}
\end{figure}
\begin{figure}[h]
 \centering
 \includegraphics[width=7cm]{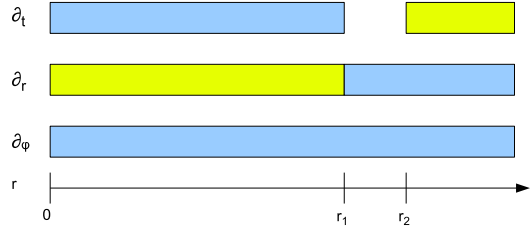}
 \caption{The causal behaviour of the coordinate vector fields of the Myers-Perry metric for $d\geq5$. A yellow bar 
          indicates timelike, a blue bar indicates spacelike behaviour. Here, $\mu=r_1^{2+k}+a^2r_1^k$ and 
          $r_2:=\sqrt[2+k]{\mu}$.}
 \label{fig.MPdKausal}
\end{figure}
\subsubsection{The general Myers-Perry metric}
%
In 1986 Myers and Perry found in \cite{myersperry} a class of spacetimes which admits rotations in any 
$N=\left\lfloor \frac{d}{2}\right\rfloor$ independent rotationplanes ($d+1$ is again the dimension of the spacetime).
It is not very surprising that within this class there is a distinction between odd and even dimension number $d$.
We will start looking at the Myers-Perry metric in its full generality and then treat the special 
case $d+1=5$ with $N=2$ independent rotation planes.

We will first introduce polar coordinates for every rotation plane: For $\{x_0,x_i\}$,$i=1,...,d$, the cartesian coordinates 
of the spacetime, the rotation planes are given by $(x_{2a-1},x_{2a})=(r_a cos\varphi_a,r_a sin\varphi_a)$, for $a=1,...N$.
Is $d+1$ an odd number, we denote the residual coordinate with $\alpha$. Furthermore let $r:=\sum_{i=1}^{d}\sqrt{x_i x_i}$ 
and we define $\mu_a=\frac{r_a}{r}$ as new coordinate function. It then holds either $\sum_{a=1}^N{\mu_a^2}=1$ or 
$\sum_{a=1}^N{\mu_a^2}+\alpha^2=1$. The coordinate $\mu_a$ is not to be confused with the mass-parameter $\mu$.

For $d+1$ odd, the general Myers-Perry metric is 
\begin{eqnarray*}
g_{MP}^o &=& -dt\otimes dt\\
         & & +\sum_{i=1}^N \left((r^2+a_i^2)(d\mu_i\otimes d\mu_i+\mu_i^2d\varphi_i\otimes d\varphi_i)
             +\frac{\mu r^2}{\Pi F}(dt-a_i\mu_i^2d\varphi_i)^2+\frac{\Pi F}{\Pi-\mu r^2}dr\otimes dr\right),
\end{eqnarray*}
for $d+1$ even, the corresponding metric is
\begin{eqnarray*}
g_{MP}^e &=& -dt\otimes dt+r^2d\alpha\otimes d\alpha\\
         & & +\sum_{i=1}^N\left((r^2+a_i^2)(d\mu_i\otimes d\mu_i+\mu_i^2d\varphi_i\otimes d\varphi_i)+
             \frac{\mu r}{\Pi F}(dt-a_i\mu_i^2d\varphi_i)^2+\frac{\Pi F}{\Pi-\mu r}dr\otimes dr\right).
\end{eqnarray*}
In both formulas the functions $F=F(r,\mu_i)$ and $\Pi=\Pi(r)$ are defined in the following way:
\begin{eqnarray*}
F(r,\mu_i)=1-\frac{a_i^2\mu_i^2}{r^2+a_i^2},\ \ \ \Pi(r)=\Pi_{i=1}^N(r^2+a_i^2).
\end{eqnarray*}
The integration constants $\mu>0$ and $a_i$ can again be associated with the mass of the rotating object and the particular 
angular momentum respectively. Note additionally that $dr$ und $d\mu_i$ aren't linearly independent, because $\mu_i=\frac{r_i}{r}$.
The both first terms of the big sum describe the behaviour of the metric on the rotation planes. Because of the fact that
the function $\Pi$ is in the denominator of the second term, it seems, as if the metric restricted on one rotationplane is 
not independent of the rotational behaviour on the other planes. The roots of the last term are coordinate singularities 
\cite{myersperry}, and for them we again want to bring up the name "horizons".

The vector fields $\partial_t$ and $\partial_{\varphi_i}$ are Killingfields, which means that the Myers-Perry solutions 
are invariant under timetranslations and under rotations along the integral curves of $\partial_{\varphi_i}$. These symmetries
build an isometry group isomorphic to $\R\times\U(1)^N$. Reducing the rotations to just one plane, we see that $g_{MP1}$ 
possesses an $\R\times\U(1)\times\SO(d-2)$ symmetry. For further discussions about the symmetries of the Myers-Perry metric, 
see \cite{emprl}. One can show \cite{myersperry} that one obtains the Kerr metric setting $d=3$. Is $a_i=0$ for all $i$ 
except for one, one can find appropriate coordinate transformations, such that the general solution $g_{MP}$ reduces to $g_{MP1}$.

\subsubsection{Horizons in Myers-Perry spacetime}
%

In the above coordinates the components of the Myers-Perry metric are singular exactly for those values of $r$ for which 
$\frac{\mu r^2}{\Pi F}=\infty$, or $\frac{\mu r}{\Pi F}=\infty$ and $\frac{\Pi F}{\Pi-\mu r^2}=\infty$, or 
$\frac{\Pi F}{\Pi-\mu r}=\infty$. The first of each case are exactly the curvature singularities, which will not be discussed  
here. For further information on that aspect see \cite{myersperry}. In this subsection we want to study the horizons of the 
Myers-Perry spacetime, which are again given by the roots of the denominator of the $rr$-components of the metric that means 
by the equation 
\bdm
\Delta_{MP}^e:=\Pi-\mu r^2=0,
\edm
if $d$ is even. We won't consider the case where $d$ is odd.\\
Because $\Pi=\Pi_{i=1}^N(r^2+a_i^2)$, the left hand side of the above equation is a polynomial of degree $d$ in $r$ and it is 
therefore not solvable with the help of a general formula. The question is now which conditions have to be fullfilled by
the $a_i^2$ to admit a horizon. A first general statement comes from the following lemma. Henceforth let $X_i:=a_i^2$.
\begin{lem}
There exists no value of $r$ for which every value of $a_i^2$ admits a horizon.
\end{lem}
\begin{proof}
Let $S_i\in\R[X_1,...,X_N]$ denote the elementary symmetric polynomials in $X_1,...,X_N$. Then we have                  
\begin{eqnarray*}
 \Pi-\mu r^2 &=& \prod_{i=1}^N(r^2+X_i)-\mu r^2 = r^{2N}+r^{2(N-1)}(\underbrace{X_1+...+X_N}_{=:S_1})\\
 & &  + r^{2(N-2)}(\underbrace{X_1X_2+X_1X_3+...+X_{N-1}X_N}_{=:S_2})+...+\underbrace{X_1\cdot...\cdot X_N}_{=:S_N}-\mu r^2\\
 &=& r^{2N}+r^{2(N-1)}S_1+r^{2(N-2)}S_2+...+S_N-\mu r^2.
\end{eqnarray*}                                     
Defining $g:=r^{2N}+r^{2(N-1)}Y_1+r^{2(N-2)}Y_2+...+Y_N-\mu r^2 \in \left(\R[Y_1,...,Y_N]\right)[\mu,r]$, it holds 
$g(S_1,...,S_N)=\Pi-\mu r^2$. Assuming there is a $(r,\mu)\in\R^2$ such that $\Pi-\mu r^2=0$, then it would be 
$g(S_1,...,S_N)=0$. This cannot be, because the elementary symmetric polynomials are algebraically independent over $\R$.
\end{proof}
This statement doesn't seem to be very surprising, especially as we could make the same statement for all the other spacetimes 
we discussed before with more elementary calculations. Nevertheless, there are examples for polynomials which are algebraically 
independent, or equivalently for general polynomial expressions that become the zero polynomial after choosing the coefficients 
appropriately. A simple example is $p(X)=aX-a^2X$. Choosing $a=1$ yields $p(X)=0$.

Whether there exists a horizon at $r=r_0$ thus depends on the choice of the $a_i$. For the sake of simplification, we want 
to analyze this dependence only for the case $d=4$. It then holds
\bdm
\Delta_{MP}^e=r^4+r^2(X_1+X_2-\mu)+X_1X_2,
\edm
which is a quadratic polynomial in $r^2$ and is therefore easily solvable: The zeros are 
\bdm
2r_{1,2}^2=\mu-X_1-X_2\pm\sqrt{(\mu-X_1-X_2)^2-4X_1X_2}.
\edm
For these solutions to be real that means for horizons being possible, the condition
\bdm
\mu\geq a_1^2+a_2^2+2|a_1a_2|=(|a_1|+|a_2|)^2
\edm
has to be fullfilled. The allowed values for the angular momenta are thus bounded and lie within a rhombus (compare Figure 
\ref{fig.phase}). Is this condition fullfilled, two horizons exist because of $\mu-X_1-X_2>\sqrt{(\mu-X_1-X_2)^2-4X_1X_2}$, 
for positive $X_1$ and $X_2$. If there is no rotation in one plane that means, if $X_i=0$ for one $i$, one obtains the zeros 
of $\Delta_{MP1}$.
\begin{figure}[h]
 \centering
 \includegraphics[width=6cm]{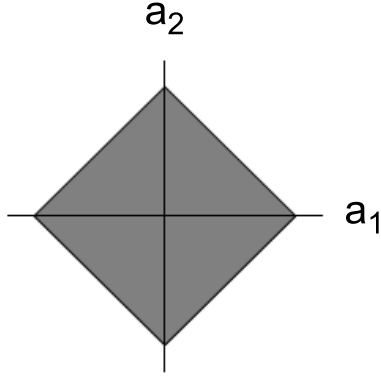}
 \caption{Sketch of the phase space of a fivedimensional rotating Myers-Perry black hole.}
 \label{fig.phase}
\end{figure}

The following lemma makes a geometric statement about the dependence of the zeroes on the $X_i$.
\begin{lem}\label{lem.quadrik}
The set of all $(X_1,X_2)\in\R^2$ wich satisfy the equation $\Delta=(r^2+X_1)(r^2+X_2)-\mu r^2=0$ for a given 
$\mu,r\in\R\setminus\{0\}$, form a hyperbola. For $r=0$ or $\mu=0$, this set forms an intersecting pair of straight lines.
\end{lem}
\begin{rem}
In conrary to a rhombus, hyperbolas are unbounded. This unboundedness originates from the fact that we also admitted negative
values of $X_i=a_i^2$ in the lemma. 
\end{rem}
\begin{proof}
\begin{eqnarray*}
(r^2+X_1)(r^2+X_2)-\mu r^2 &=& r^2(X_1+X_2)+X_1X_2+r^2(r^2-\mu)\\
                           &=& \begin{pmatrix} X_1&X_2 \end{pmatrix}
                               \underbrace{\begin{pmatrix} 0&\frac{1}{2}\\ \frac{1}{2}&0 \end{pmatrix}}_{=:A}                               
                               \begin{pmatrix} X_1\\X_2 \end{pmatrix} +
                               \begin{pmatrix} r^2&r^2  \end{pmatrix}
                               \begin{pmatrix} X_1\\X_2 \end{pmatrix} +r^2(r^2-\mu)                               
\end{eqnarray*}
We will now bring this quadric into normal form.

\emph{Step One: Determination of the eigenvalues and the corresponding eigenspaces.}\\
$det(A-\lambda Id)=\lambda^2-\frac{1}{4}=0 \Longleftrightarrow \lambda=\pm\frac{1}{2}$.\\
$(A-\frac{1}{2}Id)X=0\Longleftrightarrow X\in \{(X_1,X_2)\in\R^2:\ X_1=X_2\}$.\\
$(A+\frac{1}{2}Id)X=0\Longleftrightarrow X\in \{(X_1,X_2)\in\R^2:\ X_1=-X_2\}$.\\
The eigenspaces to the both eigenvalues $\lambda=\pm\frac{1}{2}$ of $A$ are therewith
\begin{eqnarray*}
Eig_A\left(\frac{1}{2}\right) &=& span\left\{\frac{1}{\sqrt{2}}\begin{pmatrix}1\\1\end{pmatrix}\right\},\\
Eig_A\left(-\frac{1}{2}\right)&=& span\left\{\frac{1}{\sqrt{2}}\begin{pmatrix}1\\-1\end{pmatrix}\right\}.
\end{eqnarray*}
With the help of the both stated eigenvectors, one obtains the matrix for the change of the basis
\bdm
B=
\begin{pmatrix}
\frac{1}{\sqrt{2}}&\frac{1}{\sqrt{2}}\\
\frac{1}{\sqrt{2}}&-\frac{1}{\sqrt{2}}
\end{pmatrix}.
\edm

\emph{Step Two: Transformation of the quadric with respect to the new basis.}\\
In the new basis, $A$ has the form
\bdm
B^TAB=\begin{pmatrix}\frac{1}{2}&0\\0&-\frac{1}{2}\end{pmatrix}.
\edm
With this in the new basis, the quadric has the form
\begin{eqnarray*}
0 &=& \begin{pmatrix} Y_1&Y_2 \end{pmatrix}\begin{pmatrix}\frac{1}{2}&0\\0&-\frac{1}{2}\end{pmatrix}
      \begin{pmatrix} Y_1\\Y_2 \end{pmatrix}+\begin{pmatrix}\sqrt{2}r^2&0\end{pmatrix}\begin{pmatrix}Y_1\\Y_2\end{pmatrix}
      +r^2(r^2-\mu)\\
  &=& \frac{1}{2}Y_1^2-\frac{1}{2}Y_2^2+\sqrt{2}r^2Y_1+r^2(r^2-\mu).
\end{eqnarray*}
\emph{Step Three: Translation of the origin.}\\
We finally perform the substitution 
\bdm
\begin{pmatrix} Z_1\\Z_2 \end{pmatrix} =
\begin{pmatrix} Y_1+\sqrt{2}r^2\\ Y_2 \end{pmatrix}
\edm
and obtain for the quadric the equation
\bdm
\frac{1}{2} Z_1^2-\frac{1}{2}Z_2^2-r^2\mu=0,
\edm
from which the statement follows.
\end{proof}
\begin{rem}
For the previously discussed metrics we excluded negative values for the mass parameter, because otherwise the spacetimes 
would have had naked curvature singularities, which we wanted to exclude because of the cosmic censorship hypothesis. This
hypothesis says that no naked singularities exisct, except for the bigbang singularity. But surprisingly in the case of two 
rotation planes in fivedimensional Myers-Perry spacetime, there exist horizons at $r>0$ for negative values of $\mu$.
\end{rem}
\begin{cor}
Using the form of the horizon function of the previous proof, one gains a more elegant formula for the roots:
\bdm
r=\pm\frac{1}{\sqrt{2\mu}}\sqrt{Z_1^2-Z_2^2}.
\edm
In particular, the condition $Z_1^2>Z_2^2$ is necessary for the existence of a horizon at $r\neq0$.
\end{cor}
\begin{rem}
Choosing $\mu$ to be negative in the previous corollary, a real solution is produced by extracting the factor $\sqrt{-1}$
out of $\sqrt{Z_1^2-Z_2^2}$ and demanding $Z_2^2>Z_1^2$.
\end{rem}
\subsection{Discussion of the horizon functions}\label{subsec.horizon}
To every treated metric we could associate a horizon function $\Delta$, which defined a hypersurface with special features.
This subsection is dedicated to the comparison of these important functions. As a reminder and for the sake of an overview, 
we will first list all the horizon functions we met in this article. We will further on denote the horizon function with a 
$\Delta$, but put an index for the respective metric:
\begin{eqnarray*}
	& & \Delta_{S\ \ \ }  = 1-\frac{2m}{r}\\
	& & \Delta_{T\ \ }   = 1-\frac{\mu}{r^{d-2}}\\
	& & \Delta_{K\ \ \ }  = r^2-2mr+a^2\\
	& & \Delta_{MP1}      = r^2+a^2-\frac{\mu}{r^{d-4}}\\
	& & \Delta_{MP\ \ }^e = \prod_{i=1}^{N=d/2}{(r^2+a_i^2)-\mu r^2}\\
	& & \Delta_{MP\ \ }^o = \prod_{i=1}^{N=(d-1)/2}{(r^2+a_i^2)-\mu r}
\end{eqnarray*}
$\Delta_{MP1}$ is an obvious generalization of $\Delta_T$, which on the other hand contains $\Delta_S$ as special case.
Furthermore is $\Delta_{MP}^o$ a generalization of the horizon function of the Kerr metric. 
Is with this the connection between these functions exhausted? To answer this question, let once again be pointed out that 
the essential information isn't the function itself, but its set of roots. Now, a function is not given uniquely by its set
of roots. For example possesses the product of a function $f$ with another function which is everywhere nonzero, the same set
of roots as $f$ does. We want to call two functions which only differ from such a nonvanishing function \emph{similar} and 
use the symbol $\approx$ for that. For $r\neq0$ is therefore $\Delta\approx r^k\Delta$ for all $k\geq0$. To not change the 
$rr$-component of the metric, one can simply multiply the denominator with the same power of $r$. In this way the following 
similarities result:
\begin{eqnarray*}
	& & \Delta_{S\ \ \ \ } \approx r^2-2mr\\
	& & \Delta_{T\ \ }   \approx r^{d-1}-\mu r\ \ \ \ \ \ \ \ \ \ \ \ \ \ \ \ \ \approx r^d-\mu r^2\\
	& & \Delta_{K\ \ \ }  =       r^2+a^2-2mr\\
	& & \Delta_{MP1}      \approx r^{d-1}+r^{d-3}a^2-\mu r^2\ \ \ \approx r^d+r^{d-2}a^2-\mu r^2\\
	& & \Delta_{MP\ \ }^e =       \prod_{i=1}^{N=d/2}{(r^2+a_i^2)-\mu r^2}\\
	& & \Delta_{MP\ \ }^o =       \prod_{i=1}^{N=(d-1)/2}{(r^2+a_i^2)-\mu r}.
\end{eqnarray*}
Now it is possible to see more clearly the relationship between the different horizon functions. For $\Delta_T$ and $\Delta_{MP1}$ 
two similarities are given to point out the relationship to $\Delta_{MP}^e$ as well as to $\Delta_{MP}^o$.

We also want to discuss the role of $a$ or the $a_i$. Setting $a=0$, then $\Delta_K$ becomes $\Delta_S$ and $\Delta_{MP1}$ 
becomes $\Delta_T$. Setting further in $\Delta_{MP}^e$ $a_i=0$ for every $i$ but one, without loss of generality let $a_1\neq0$, 
then 
\begin{eqnarray*}
\Delta_{MP}^g &=& r^{2(N-1)}(r^2+a_1^2)-\mu r^1 = r^{d-2}(r^2+a_1^2)-\mu r^2\\
              &=& r^d+r^{d-2}a_1^2-\mu r^2 = \Delta_{MP1}.
\end{eqnarray*}
An analogous calculation can be done for $\Delta_{MP}^o$. Comparing $\Delta_{MP1}$ with $\Delta_{MP}^g$ or $\Delta_{MP}^u$, 
one realizes that for every additional rotationplane a factor $r^2$ of $\Delta_{MP1}$ is "converted" into $r^2+a_i^2$. 
In $\Delta_{MP}^g$ and $\Delta_{MP}^u$ we thus found two functions, in which every other horizon function is contained. 

By the insight, how the horizon functions are related and with Lemma \ref{lem.quadrik} we can now understand better the dependence 
of the existence of a horizon for a given $r$ and $\mu$ from the choice of the angular momenta. Lemma \ref{lem.quadrik} namely 
says that for given $r$ and $\mu$ there are infinitely many possibilities to choose such $a_1$ and $a_2$ which allow the 
existence of an horizon. This wasn't the case for metrics which considered only rotation in one plane. There, always two 
possibilities existed: 
\begin{itemize}
 \item $a=\pm\sqrt{m^2-(r-m)^2}$ in Kerr spacetime and
 \item $a=\pm\sqrt{\mu-r^2}$ in fivedimensional Myers-Perry spacetime with only one rotation plane.                          
 \end{itemize}
Setting one parameter of a hyperbola equal to zero, the remaining parameter has only two possibilities left.

Finally let us point out the remarkable fact be pointed out that the horizon functions are similar to polynomials in $r$, 
or simply are polynomials, what maybe wasn't to be expected. 
%
\newpage
\section*{Appendix: Ricci-flatness of the Tangherlini metric}
In this appendix we want to show that the Tangherlini metric is indeed Ricci-flat, as to the authors knowledge a proof of that 
fact still cannot be found in the literature. In addition, in this proof we will use the statement of Lemma \ref{lem.ricci} 
which is also supposed to be a new result.

To show that a metric fulfills the vacuum Einstein equations, it suffices to show that it is Ricci-flat. For this purpose
we use the Cartan structure formalism. Therefor we define an orthonormal basis of $1$-forms $\{\Theta^l\}$ $l=0,...,d$ by    
\begin{eqnarray*}
\Theta^0 & = & \sqrt{\left(1-\frac{\mu}{r^{d-2}}\right)}\ dt\\ 
\Theta^1 & = & \frac{1}{\sqrt{1-\frac{\mu}{r^{d-2}}}}\ dr\\
\Theta^2 & = & r\ d\chi_2\\ 
\Theta^i & = & r\prod_{s=2}^{i-1}\sin\chi_s\ d\chi_i, 
\end{eqnarray*}
where the $\{\chi_i\}$ again denote the generalized spherical coordinates and $i=3,...,d$. 
We recall that for the connection forms with respect to orthonormal bases the symmetry relations
$\omega^0_i=\omega^i_0,\ \omega^i_j=-\omega^j_i$ hold. In particular it holds $\omega^i_i=0$. With the help of these relations 
and the first structure equation $d\theta^i+\omega^i_k\wedge\Theta^k=0$ the connection forms are able to be uniquely determined.
For this purpose we firstly calculate the total differential of the above $1$-forms:
\begin{eqnarray*}
d\Theta^0 &=& \frac{(d-2)\mu}{2r^{d-1}}\frac{1}{\sqrt{1-\frac{\mu}{r^{d-2}}}}\ dr\wedge dt =                  \frac{(d-2)\mu}{2r^{d-1}}\frac{1}{\sqrt{1-\frac{\mu}{r^{d-2}}}}\ \Theta^1\wedge\Theta^0\\
d\Theta^1 &=& 0\\
d\Theta^2 &=& dr\wedge d\chi_2 = \frac{1}{r}\sqrt{1-\frac{\mu}{r^{d-2}}}\ \Theta^1\wedge\Theta^2\\
d\Theta^i &=& \prod_{s=2}^{i-1}\sin\chi_s\ dr\wedge d\chi_i+r\sum_{k=2}^{i-1}\left(\cos\chi_k\prod_{s=2,s\neq k}^{i-1}
\sin\chi_s\ d\chi_k\wedge d\chi_i\right)\\
          &=& \frac{1}{r}\sqrt{1-\frac{\mu}{r^{d-2}}}\ \Theta^1\wedge\Theta^i+\frac{1}{r}\sum_{k=2}^{i-1}\left(\cot\chi_k
\prod_{s=2}^{k-1}\frac{1}{\sin\chi_s}\ \Theta^k\wedge\Theta^i\right).
\end{eqnarray*}
For $i>2$. For the empty product we set $\prod_{k=2}^{1}\frac{1}{\sin\chi_k}:=1$.

After comparison to the first structure equation the connection forms which are different from zero yield
\begin{eqnarray*}
\omega^0_1 &=& \omega^1_0 = \frac{(d-2)\mu}{2r^{d-1}}\frac{1}{\sqrt{1-\frac{\mu}{r^{d-2}}}}\ \Theta^0\\
\omega^2_1 &=& -\omega^1_2 = \frac{1}{r}\sqrt{1-\frac{\mu}{r^{d-2}}}\ \Theta^2\\
\omega^i_1 &=& -\omega^1_i = \frac{1}{r}\sqrt{1-\frac{\mu}{r^{d-2}}}\ \Theta^i\\
\omega^i_l &=& -\omega^l_i = \frac{\cot\chi_l}{r}\prod_{s=2}^{l-1}\frac{1}{\sin\chi_s}\ \Theta^i,
\end{eqnarray*}
where $2\leq l\leq i-1$ und $i>2$. With the usage of the second structure equation $d\omega^i_j+\omega^i_k\wedge\omega^k_j=\Omega^i_j$
one now can calculate the curvature forms $\Omega^i_j$. For this we first calculate total differentials of the connection 
forms:
\begin{eqnarray*}
d\omega^0_1 &=& d\left(\frac{(d-2)\mu}{2r^{d-1}}\ dt\right) = -\frac{(d-1)(d-2)\mu}{2r^{d}}\ dr\wedge dt 
             = -\frac{(d-1)(d-2)\mu}{2r^{d}}\ \Theta^1\wedge\Theta^0\\
d\omega^2_1 &=& d\left(\sqrt{1-\frac{\mu}{r^{d-2}}}\ d\chi_2\right) = \frac{(d-2)\mu}{2r^{d-1}}\frac{1}{\sqrt{1-\frac{\mu}{r^{d-2}}}}\ dr\wedge d\chi_2
             = \frac{(d-2)\mu}{2r^{d}}\ \Theta^1\wedge\Theta^2\\
d\omega^i_1 &=& d\left(\sqrt{1-\frac{\mu}{r^{d-2}}}\prod_{s=2}^{i-1}\sin\chi_s\ d\chi_i\right)\\
            &=& \frac{(d-2)\mu}{2r^{d-1}}\frac{1}{\sqrt{1-\frac{\mu}{r^{d-2}}}}\prod_{s=2}^{i-1}\sin\chi_s\ dr\wedge d\chi_i
             + \sqrt{1-\frac{\mu}{r^{d-2}}}\sum_{k=2}^{i-1}\left(\cot\chi_k\prod_{s=2}^{i-1}\sin\chi_s\ d\chi_k\wedge d\chi_i\right)\\
            &=& \frac{(d-2)\mu}{2r^{d}}\ \Theta^1\wedge\Theta^i+ \frac{1}{r^2}\sqrt{1-\frac{\mu}{r^{d-2}}}\sum_{k=2}^{i-1}\left(\cot\chi_k\prod_{s=2}^{k-1}\frac{1}{\sin\chi_s}\ \Theta^k\wedge\Theta^i\right)\\
d\omega^i_l &=& d\left(\cos\chi_l\prod_{s=l+1}^{i-1}\sin\chi_s\ d\chi_i\right)\\
            &=& -\prod_{s=l}^{i-1}\sin\chi_s\ d\chi_l\wedge d\chi_i
             + \sum_{k=l+1}^{i-1}\cot\chi_l\cot\chi_k\prod_{s=l}^{i-1}\sin\chi_s\ d\chi_k\wedge d\chi_i\\
            &=& -\frac{1}{r^2}\prod_{s=2}^{l-1}\frac{1}{sin^2\chi_s}\ \Theta^l\wedge\Theta^i
             + \frac{1}{r^2}\cot\chi_l\prod_{s=2}^{l-1}\frac{1}{\sin\chi_s}\sum_{k=l+1}^{i-1}\left(\cot\chi_k\prod_{s'=2}^{k-1}\frac{1}{\sin\chi_s'}\ 
                \Theta^k\wedge\Theta^i\right).
\end{eqnarray*}
Where again holds $i>j$. We now plug in the found expressions into the second structure equation. The curvature forms which are 
different from zero then yield as follows. Thereby is $i,k>1$, $j>2$, $l>j$ and the relations $\Omega^0_i=\Omega^i_0$ and 
$\Omega^i_k=-\Omega^k_i$ hold.
\begin{eqnarray*}
\Omega^0_1 &=& d\omega^0_1+\omega^0_k\wedge\omega^k_1= d\omega^0_1= -\frac{(d-1)(d-2)\mu}{2r^{d}}\ \Theta^1\wedge\Theta^0\\
\Omega^0_i &=& d\omega^0_i+\omega^0_k\wedge\omega^k_i= \omega^0_1\wedge\omega^1_i= -\frac{(d-2)\mu}{2r^{d}}\ \Theta^0\wedge\Theta^i\\
\Omega^1_2 &=& d\omega^1_2+\omega^1_k\wedge\omega^k_2 = d\omega^1_2 = -\frac{(d-2)\mu}{2r^{d}}\ \Theta^1\wedge\Theta^2\\
\Omega^1_j &=& d\omega^1_j+\omega^1_k\wedge\omega^k_j = -\frac{(d-2)\mu}{2r^{d}}\ \Theta^1\wedge\Theta^j
               - \frac{1}{r^2}\sqrt{1-\frac{\mu}{r^{d-2}}}\sum_{k=2}^{j-1}\left(\cot\chi_k
               \prod_{s=2}^{k-1}\frac{1}{\sin\chi_s}\ \Theta^k\wedge\Theta^j\right)\\
           & & +\frac{1}{r^2}\sqrt{1-\frac{\mu}{r^{d-2}}}\sum_{k=2}^{j-1}\left(\cot\chi_k\prod_{s=2}^{k-1}\frac{1}{\sin\chi_s}\ 
               \Theta^j\wedge\Theta^k\right)
            = -\frac{(d-2)\mu}{2r^{d}}\ \Theta^1\wedge\Theta^j\\
\Omega^2_j &=& d\omega^2_j+\omega^2_k\wedge^k_j = \frac{1}{r^2}\ \Theta^2\wedge\Theta^j
               - \frac{1}{r^2}\cot\chi_2\sum_{k=3}^{j-1}\left(\cot\chi_k\prod_{s=2}^{k-1}\frac{1}{\sin\chi_s}\ \Theta^k\wedge\Theta^j\right)\\
           & & - \frac{1}{r^2}\left(1-\frac{\mu}{r^{d-2}}\right)\ \Theta^2\wedge\Theta^j
               + \frac{1}{r^2}\cot\chi_2\sum_{k=3}^{j-1}\left(\cot\chi_k\prod_{s=2}^{k-1}\frac{1}{\sin\chi_s}\ \Theta^k\wedge\Theta^j\right)\\
           &=& \frac{\mu}{r^{d}}\ \Theta^2\wedge\Theta^j
\end{eqnarray*}
\begin{eqnarray*}            
\Omega^j_l &=& d\omega^j_l+\omega^j_1\wedge\omega^1_l+\sum_{k=2}^{j-1}\omega^j_k\wedge\omega^k_l+\sum_{k=j+1}^{l-1}\omega_k^j\wedge\omega_l^k\\
           &=& \frac{1}{r^2}\prod_{s=2}^{j-1}\frac{1}{\sin^2\chi_s}\ \Theta^j\wedge\Theta^l
               -\frac{1}{r^2}\cot\chi_j\prod_{s=2}^{j-1}\frac{1}{\sin\chi_s}\sum_{k=j+1}^{l-1}\left(\cot\chi_k\prod_{s'=2}^{k-1}\frac{1}{\sin\chi_s'}
               \ \Theta^k\wedge\Theta^l\right)\\
           & & -\frac{1}{r^2}\left(1-\frac{\mu}{r^{d-2}}\right)\ \Theta^j\wedge\Theta^l
               -\frac{1}{r^2}\left(\sum_{k=2}^{l-1}\cot^2\chi_k\prod_{s=2}^{k-1}\frac{1}{\sin^2\chi_s}\right)\ \Theta^j\wedge\Theta^l\\
           & & +\frac{1}{r^2}\cot\chi_j\prod_{s=2}^{j-1}\frac{1}{\sin\chi_s}\sum_{k=j+1}^{l-1}\left(\cot\chi_k\prod_{s'=2}^{k-1}\frac{1}{\sin\chi_s'}
               \ \Theta^k\wedge\Theta^l\right)\\
           &=& \frac{1}{r^2}\left(\prod_{s=2}^{j-1}\frac{1}{\sin^2\chi_s}
              -\sum_{k=2}^{j-1}\left(\cot^2\chi_k\prod_{s=2}^{k-1}\frac{1}{\sin^2\chi_s}\right)+\frac{\mu}{r^{d-3}}-1\right)\ \Theta^j\wedge\Theta^l \\
           &=& \frac{\mu}{r^{d}}\ \Theta^j\wedge\Theta^l.
\end{eqnarray*}
For the calculation of the $\Omega^i_j$ the following Lemma was used:

\begin{lem}\label{lem.ricci}                       
\bdm
\prod_{s=2}^{n-1}\frac{1}{\sin^2\chi_s}-\sum_{k=2}^{n-1}\left(\cot^2\chi_k\prod_{s=2}^{k-1}\frac{1}{\sin^2\chi_s}\right)=1\ \ \forall n>2
\edm
\end{lem}
\begin{proof}
We perform an induktion with respect to $n$. At first we consider the case $n$=3:
\bdm
\prod_{s=2}^2\frac{1}{\sin^2\chi_s}-\sum_{k=2}^{2}\left(\cot^2\chi_k\prod_{s=2}^1\frac{1}{\sin^2\chi_s}\right)
= \frac{1}{\sin^2\chi_2}-\cot^2\chi_2 = \frac{1}{\sin^2\chi_2}\left(1-\cos^2\chi_2\right)=1
\edm
Does the statement of the Lemma hold for $n-1>2$, then it also holds for $n$:
\begin{eqnarray*}
& &\prod_{s=2}^{n-1}\frac{1}{\sin^2\chi_s}-\sum_{k=2}^{n-1}\left(\cot^2\chi_k\prod_{s=2}^{k-1}\frac{1}{\sin^2\chi_s}\right)
   =\prod_{s=2}^{n-1}\frac{1}{\sin^2\chi_s}\left(1-\sum_{k=2}^{n-1}\left(\cot^2\chi_k\prod_{s=k}^{n-1}\sin^2\chi_s\right)\right)\\
&=& \prod_{s=2}^{n-1}\frac{1}{\sin^2\chi_s}\left(1-\cos^2\chi_{n-1}-\sum_{k=2}^{n-2}\left(\cot^2\chi_k\prod_{s=k}^{n-1}\sin^2\chi_s\right)\right)\\
&=& \prod_{s=2}^{n-2}\frac{1}{\sin^2\chi_s}\left(1-\sum_{k=2}^{n-2}\left(\cot^2\chi_k\prod_{s=k}^{n-1}\sin^2\chi_s\right)\right)
   =\prod_{s=2}^{n-2}\frac{1}{\sin^2\chi_s}-\sum_{k=2}^{n-2}\left(\cot^2\chi_k\prod_{s=2}^{k-1}\frac{1}{\sin^2\chi_s}\right)\\
&=& 1.
\end{eqnarray*}
\end{proof}
To now calculate the components of the Ricci-tensor, $R_{\mu\nu}$, we now use the relation $R_{\mu\nu}=\Omega_\mu^\alpha(e_\alpha,e_\nu)$,
where $\{e_i\}$ is the basis dual to $\{\Theta^i\}$, $i=0,...,d$.
Because of $\Omega_\mu^\alpha\sim\Theta^\mu\wedge\Theta^\alpha$ it follows $\Omega_\mu^\alpha(e_\alpha,e_\nu)=0$ for $\mu\neq\nu$. 
This is why $R_{\mu\nu}=0$ for $\mu\neq\nu$. We calculate the value of the remaining components. For this let be $j>2$:

\begin{eqnarray*}
R_{00} &=& -\frac{(d-1)(d-2)\mu}{2r^{d}}+\sum_{i=2}^{d+1}\frac{(d-1)\mu}{2r^{d}}=0\\
R_{11} &=&  \frac{(d-1)(d-2)\mu}{2r^{d}}-\sum_{i=2}^{d+1}\frac{(d-2)\mu}{2r^{d}}=0\\
R_{22} &=& -\frac{(d-2)\mu}{2r^{d}}-\frac{(d-2)\mu}{2r^{d}}+\sum_{j=3}^{d+1}\frac{\mu}{r^{d}}=0\\
R_{jj} &=& -\frac{(d-2)\mu}{2r^{d}}-\frac{(d-2)\mu}{2r^{d}}+\frac{\mu}{r^{d}}+\sum_{j=4}^{d+1}\frac{\mu}{r^{d}}=0.
\end{eqnarray*}

Thus, the Tangherlini metric is Ricci-flat.

\nocite{*}
\bibliographystyle{amsalpha}
\bibliography{References}

\end{document}